\documentclass[final,5p,times,twocolumn,number,sort&compress]{elsarticle} 

\usepackage{graphicx}
\usepackage{epstopdf}
\usepackage{color}
\usepackage{url}
\usepackage{booktabs} 
\usepackage{amsmath} 
\usepackage{amssymb}
\usepackage{amsthm}
\usepackage{eurosym} 
\usepackage{todonotes}
\usepackage{pstool}
\usepackage{siunitx}
\usepackage{fancyhdr}
\newtheorem{thm}{Theorem}
\newtheorem{definition}{Definition}
\newproof{pf}{Proof}

\newcommand{\startmodif}{\color{black}}
\newcommand{\stopmodif}{\color{black}}

\begin{document}

\begin{frontmatter}

\title{\LARGE Least costly energy management for series hybrid electric vehicles}

\author{Simone~Formentin\corref{cor1}}\ead{simone.formentin@polimi.it}
\author{Jacopo~Guanetti\corref{cor2}}
\author{Sergio~M.~Savaresi\corref{cor2}}
\cortext[cor1]{Corresponding author. Tel: +39 02 2399 3498 - Fax: +39 02 2399 3412}
\address{Dipartimento di Elettronica, Informazione e Bioingegneria, Politecnico di Milano, P.za L. da Vinci 32, 20133 Milano, Italy.}

\begin{abstract}
Energy management of plug-in Hybrid Electric Vehicles (HEVs) has different challenges from non-plug-in HEVs, due to bigger batteries and grid recharging. Instead of tackling it to pursue energetic efficiency, an approach minimizing the driving cost incurred by the user -- the combined costs of fuel, grid energy and battery degradation -- is here proposed. A real-time approximation of the resulting optimal policy is then provided, as well as some analytic insight into its dependence on the system parameters. The advantages of the proposed formulation and the effectiveness of the real-time strategy are shown by means of a thorough simulation campaign.
\end{abstract}


\end{frontmatter}                      


\section{Introduction}

Hybrid Electric Vehicles (HEVs) are generally regarded to as an effective solution to improve the fuel economy and reduce CO2 emissions with respect to Internal Combustion Engine (ICE) vehicles. 
Since HEVs are usually equipped with (at least) two energy sources, a critical energy management problem arises, that is, a supervisory system is needed to determine how to generate the requested power. In the so-called ``mild HEVs", the downsized battery and the electrical motor do not allow to drive the vehicle based just on the electric power, but only to assist the ICE in low efficiency operating points. In this framework, heuristics and rule-based algorithms have shown to provide satisfactory results. On the other hand, highly hybridized powertrains call for more sophisticated control approaches for their higher flexibility \cite{sciarretta_control_2007}.

In the latter configuration, given a model of the hybrid powertrain, the best performance theoretically achievable over a driving schedule can be computed by means of optimization techniques, see, \textit{e.g.}, \cite{delprat2004control,barsali2004control}. A classical approach in HEVs aims at minimizing the overall fuel consumption, concurrently penalizing excessive deviations of the battery state of charge \cite{won2005energy,won2005intelligent}. Such a penalty term is very important for conventional HEVs, in which the minimization of the fuel consumption \textit{tout court} may lead to excessive battery charge depletion . 

The above optimization approach usually yields a non-causal control policy, which defines a useful upper bound in terms of performance for a given driving cycle. A good approximation of the above optimal policy can be found using the so-called Equivalent Consumption Minimization Strategy (ECMS) - based on the Pontryagin Minimum Principle - in which the knowledge of future power requests is replaced by a cycle-dependent parameter, see  \cite{sciarretta_control_2007,sciarretta_optimal_2004,serrao_ecms_2009,paganelli2000simulation,Kim2011}  for further details. Adaptive variants of the ECMS have also been developed and successfully implemented in real-time \cite{musardo_aecms_2005,ambuhl2009predictive}. Nonetheless, other real time approaches have been explored, based, \textit{e.g.}, on Model Predictive Control \cite{borhan2012mpc,poramapojana2012predictive} or Robust Control \cite{pisu2007comparative,pisu_lmi-based_2003}. 

The above strategies were originally conceived for conventional, non plug-in HEV powertrains, that is when the battery can be recharged exclusively during vehicle operation, \emph{e.g.} by regenerative braking or thermal power surplus. However, more recent plug-in HEVs make it possible to recharge the battery from the grid \cite{bradley2009design,axsen2013hybrid}. Quite simultaneously, progresses in battery technology are making big battery packs more affordable, thus extending the electric autonomy of such vehicles.

Upcoming HEVs are then more and more conceived as plug-in vehicles with a relatively large battery and a significant ``all-electric range", with a thermal unit often playing the role of a range extender. In view of this trend, on the one hand, the need for charge sustenance becomes less critical. On the other hand, since the battery has a more significant impact on the overall vehicle cost, the battery operating conditions leading to fast aging should be avoided.

Supervisory strategies have been proposed also for the energy management of plug-in and series HEVs. Sticking as a relevant case to the ECMS strategies mentioned above, some implementations for a plug-in HEV are presented \textit{e.g.} in \cite{Sciarretta2014}; quite intuitively, here the charge sustenance constraint can be relaxed, by taking into account the characteristics of the powertrain and the available information on the trip to be performed. A general framework for energy optimization of plug-in HEVs has been recently introduced in \cite{guardiola2014insight}, where the optimal data-driven tuning of the ECMS policy is also discussed.
In some recent works, battery aging is accounted for in the optimization problem. In \cite{moura2013battery} battery aging and energy consumption are both regarded as relevant phenomena for the optimal depletion strategy of the battery in a plug-in HEV. In \cite{Serrao2011,Ebbesen2012} a similar problem is tackled for HEVs with a hard charge sustenance constraint; in these works, ECMS-based strategies are developed, with an additional tuning parameter affecting the weight of the aging in the cost function.

\startmodif
The contributions of this paper can be summarized as follows.
Firstly, a \emph{least costly} formulation of the energy management is proposed, aiming to fully exploit series hybrid powertrains.
The underlying model also accounts for battery aging and the optimal control problem accounts for all the \emph{cost entries} related to both the electrical part and the thermal unit.

Secondly, by applying Dynamic Programming (DP) \cite{bertsekas1995dynamic}, it is shown that the resulting energy management policy does not necessarily yield minimum fuel consumption.
As a matter of fact, cheap fuels like CNG (Compressed Natural Gas) can prove cheaper than driving entirely on electric power, especially if battery purchase cost is considered; in such a scenario, a formulation in terms of \textit{total driving cost} is desirable from the point of view of the user.
Moreover, limited diffusion of alternative fuels may boost the adoption of multi-fuel range-extenders \cite{Fuerex}.
In the latter case, a total driving cost formulation allows to find a compromise \textit{e.g.} between a relatively expensive fuel that is easy to find, like gasoline, and a cheaper less widespread fuel, like CNG.

Unfortunately, the above DP-based solution relies upon the \textit{a-priori} knowledge of the driving cycle.
Therefore, as a further contribution of the paper, two causal implementations of the least costly energy management strategy are proposed.
The optimal policy is first derived based on a simplified model of the powertrain in an explicit way: although the model is less general, in this case the policy is expressed as a set of explicit rules, hence its implementation requires substantially less memory and computational power.
Furthermore it is shown that, when a more complex model of the powertrain is necessary, the optimal policy can still be computed numerically, attaining very close results to the acausal benchmark. Finally, the paper includes a sensitivity analysis, that investigates the performance of the numerical policy for a broad range of model parameters and energy costs.
\stopmodif

The remainder of the paper is as follows. A general formulation of the energy management problem - as well as some specific formulations in terms of energy consumption minimization - is given in \emph{Section \ref{sec:Section2}}, where the full-fledged simulator of the vehicle and the simulation scenarios used in the following sections are also presented. By deriving a suitable control-oriented model and an economic cost function, the least costly energy management approach is presented in \emph{Section \ref{sec:Section3}}, where the resulting non-causal policy is also derived by Dynamic Programming. \emph{Section \ref{sec:Section4}} provides the causal policies for the least costly energy management problem, while \emph{Section \ref{sec:Section5}} discusses the limits of applicability of such a strategy by means of a sensitivity study. The potential of the new approach is shown in each section by employing both a urban and a mixed urban-motorway driving cycle. The paper is ended by some concluding remarks.


\section{Problem formulation and simulation setup}
\label{sec:Section2}

In this section the HEV energy management problem is presented and the way it is commonly addressed in the literature is discussed. Moreover, the simulation setup and the driving cycles -- employed in the remainder of the paper to test the proposed strategy -- are introduced.

\subsection{Problem Formulation}
With ``energy management problem" it is meant the problem of designing a supervisory control layer with the aim of managing the power dispatch between multiple sources in a HEV.
More specifically, such a problem is commonly formalized as an optimal control problem over a finite time horizon.
With reasonable knowledge of the vehicle, the speed and slope profiles of the trip can be converted into a profile of requested electrical power in series HEVs, or mechanical power in parallel HEVs. The remainder of the paper is focused on series HEVs.

Formally, an energy management problem can be written as
\begin{equation}
\label{eq:GenericOCP}
\begin{aligned}
& \underset{u}{\text{min}}
& & J = h(x(T)) + \int_0^T g(t,x(t),u(t),w(t))dt \\
& \text{s. t.}
& & \dot{x} = f(t,x(t),u(t),w(t)) \\
&
& & x(0) = x_0\\
&
& & x(t) \in X\\
&
& & x(T) \in X_T\\
&
& & u(t) \in U
\end{aligned}
\end{equation}
where $J$ is the cost function to minimize, 
$x$ collects the state variables, 
$u$ is the control variable, 
$w$ represents the exogenous input variable,
$f$ denotes the state function,
$g$ is the running cost and
$h$ is the terminal cost.
$x, u, w$ are assumed to be scalar variables and $f, g, h$ are assumed to be scalar, possibly nonlinear functions.
$X=[\bar{x}_{min},\bar{x}_{max}] \subseteq \mathbb{R} $ is the set of admissible values for the state variable; the bounds $\bar{x}_{min},\bar{x}_{max}$ are assumed to be static.
$X_T$ is the set of admissible values for the final state.
$U=[u_{min}(t),u_{max}(t)] \subset \mathbb{R} \times \mathbb{R}^{T-1}$ is the set of admissible values for the input variable; the bounds $u_{min},u_{max}$ are assumed to possibly be time-varying.

Many approaches proposed in the literature aim at minimizing the fuel consumption for a given trip; therefore, the fuel mass flow rate is often chosen as the running cost as
\begin{equation}
\label{eq:FuelCost}
g(u(t)) = \dot{m}_{f}(u(t)).
\end{equation}
The fuel mass flow rate reasonably depends on the control policy $u(t)$.
The control input may be the battery current, the battery power, the generated power or the ratio between battery and generated power.
The state variable is typically the battery state of charge, which requires the introduction of a battery model.

A possible strategy is the \textit{Full Electric} mode, \textit{i.e.} the simple minimization of the fuel consumption, without any constraint or penalization on the final state
\begin{equation}
\label{eq:FEFinalCost}
\begin{array}{rl}
h(x(T)) &= 0\\
X_T &= X.
\end{array}
\end{equation}

As it can be easily understood, such an approach can lead to excessive depleting of the battery charge, if the trip is sufficiently long or demanding.
To overcome such an issue, a constraint on the final state of charge can be set, \textit{e.g.} the desired final state can be enforced to be equal to the initial state such that perfect ``charge sustenance" is accomplished. Notice that for the class of HEVs here considered, an arbitrary final state of charge could be requested by the user; the charge sustaining case is just one possible choice. 
Formally, this means
\begin{equation}
\label{eq:CSFinalCost}
\begin{array}{rl}
h(x(T)) &= 0\\
X_T &= \lbrace x(0) \rbrace,
\end{array}
\end{equation}
and this approach will be referred to as \textit{Charge Sustaining} mode, from now on.

A hard constraint on the final state can be fulfilled when a non-causal solution of the problem is computed.
When the optimization has to be performed in real-time, soft constraints are simpler to handle. The so-called \textit{Equivalent Consumption Minimization Strategy} (ECMS) approach \cite{sciarretta_control_2007,sciarretta_optimal_2004,serrao_ecms_2009,paganelli2000simulation} has an effective real-time implementation and can be stated as
\begin{equation}
\label{eq:ECMSFinalCost}
\begin{array}{rl}
h(x(T)) &= \zeta(x(T) - x(0))\\
X_T &= X,
\end{array}
\end{equation}
where $\zeta$ is a user-defined parameter usually depending on the considered driving cycle. Notice that ECMS strategies sometimes consider a modified cost function (see e.g. \cite{ambuhl2009predictive}) that penalizes not only the fuel consumption, but also any deviation of the state of charge from a reference value $x_{ref}(t)$ 
\begin{equation}
\label{eq:FuelCostECMS}
g(u(t)) = \dot{m}_{f}(u(t)) + \mu \left( \frac{x_{ref}(t)-x(t)}{\Delta x} \right)^{2q} ,
\end{equation}
where $\mu$ is a weighting factor, $q$ denotes the order of the penalty term and $\Delta x$ is the largest admissible deviation of the state of charge.

The above formulation of the energy management problem has been widely investigated in the literature (see \textit{e.g.} \cite{sciarretta_control_2007,sciarretta_optimal_2004,serrao_ecms_2009}).
Nonetheless, it was originally conceived for parallel HEVs, not rechargeable from the grid.
Current trend in HEVs market is oriented towards grid rechargeable vehicles with significant battery capacity.
Rather than merely increasing the efficiency of the thermal engine, in many cases the purpose is also to have an electric range suited for every-day urban use; a clue is the increasing commercial offer of plug-in HEVs and Extended Range EVs.
This peculiar aspects are considered in the recent literature when studying the energy management problem for this new generation of HEVs.
As a relevant and recent example, some implementations of the aforementioned ECMS for a plug-in HEV are presented in \cite{Sciarretta2014}; consistently with the remarks made above, the charge sustenance constraint is relaxed, taking into account the significant battery capacity and the available information on the trip to be performed.
As discussed in \cite{guardiola2014insight}, the tuning of ECMS policies for plug-in HEVs is highly related to the relative weight of battery energy and fuel energy; this weighting can be formulated to minimize, e.g., the vehicle energy or the CO2 emissions.

In this work, a formulation of the energy management problem is proposed, taking into account the phenomena with greatest impact on the overall driving cost incurred by the user. This choice will make the energy management a real least costly policy.

The main motivation of such a choice is that, in recent vehicles, the electric portion of the powertrain has a very significant impact on the overall cost of the vehicle, mainly due to the battery cost \cite{shiau2010optimal}.
This calls for a different formulation of the management problem, in which also the electric costs are explicitly considered in the overall minimization criterion.
Since this study is focused on series HEVs, it is as well reasonable to assume that a battery with significant capacity is available on board; in other words, for this class of HEVs, hard charge sustenance is not a critical need.
Therefore, while the cost of both electric and fuel energy are relevant \cite{guardiola2014insight}, it is also important to include the effect of the battery value depletion.

\subsection{Simulation Setup}\label{sec:simulator}
In order to test the proposed strategies, a full-fledged simulator of a series HEV was implemented using a backward facing approach. This simulator is based on the platform developed at ETH Zurich \cite{ETHsimulator} and is available on-line \cite{POLIMIsimulator}. More details about its implementation are given in Appendix A.

Backward-facing approach is generally intended for the simulation of the overall energetic performance of the vehicle, over time scales comparable to the duration of standard driving cycles.
The subsystems are described by approximate, quasi-static or low-order dynamic models.
The name of the approach refers to the fact that the inputs to the simulator are the vehicle road performance, for instance in terms of longitudinal speed and slope.
The behavior of the vehicle is simulated proceeding backwards in the powertrain and computing the upstream energy flows.
The outputs of our simulator are the fuel consumption, the battery state of charge and the battery state of health.

A conceptual scheme of the simulator is depicted in Fig.~\ref{fig:Simulator}, where $v$ is the vehicle longitudinal speed, $\theta$ is the road slope, $q_b$ is the battery state of charge, $\xi_b$ is the battery state of health, $\dot{m}_f$ is the fuel consumption of the thermal generator.
Notice that in Fig.~\ref{fig:Simulator} thin lines represent electric connections, while thick lines represent mechanical connections.

Three main areas can be identified, a vehicle dynamics area (containing the vehicle itself, the transmission and the motor), a battery area and a thermal generation area (containing the thermal engine and the electric generator).
The three areas are connected by means of an electric power link, where the sum of the battery power $P_b$ and of the generated power $P_r$ has to be equal to the traction motor power $P_m$.
Details on the parametrization and the equations included in each block are given in the Appendix.

\begin{figure*}
\centering
\includegraphics[]{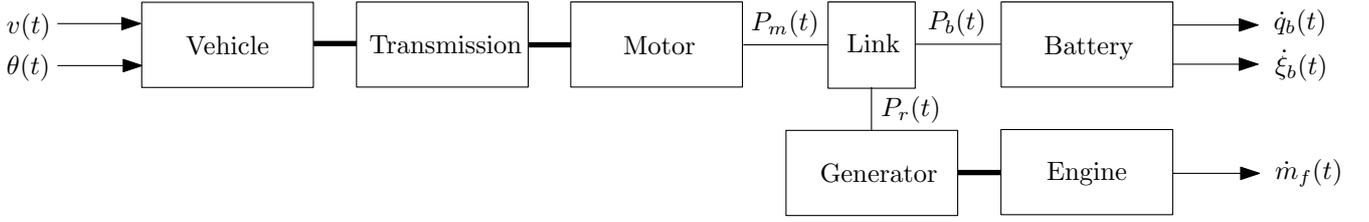}
\caption{Overall composition of the Vehicle Simulator.}
\label{fig:Simulator} 
\end{figure*}

In this work two different driving cycles are used as inputs to the simulator. The first driving cycle, depicted in Fig.~\ref{fig:UrbanCycle}, is the standard FTP urban driving cycle and is named \textit{Urban Driving Cycle} hereafter. The second driving cycle, depicted in Fig.~\ref{fig:CombinedCycle}, is obtained by merging the FTP urban driving cycle and the FTP Highway driving cycle and is named \textit{Combined Driving Cycle} hereafter.

\begin{figure}
\centering
\includegraphics[]{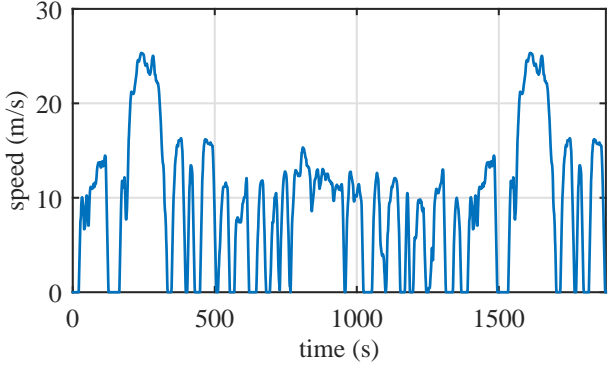}
\caption{\startmodif Urban Driving Cycle. \stopmodif}
\label{fig:UrbanCycle}
\end{figure}

For both driving cycles, a realistic scenario is considered, in which the vehicle has already depleted 50\% of the charge, but has to accomplish another trip before having the possibility to recharge the battery from the grid. Moreover, to make the trade-off between thermal and electric power less trivial, a CNG range-extender is considered.
This choice falsifies the widespread belief that ``electric is cheaper": as a matter of fact, \emph{with reasonable costs of fuel, energy and battery, the usage of a CGN engine is found to be preferable even within the range achievable in full-electric mode.} 

\begin{figure}
\centering
\includegraphics[]{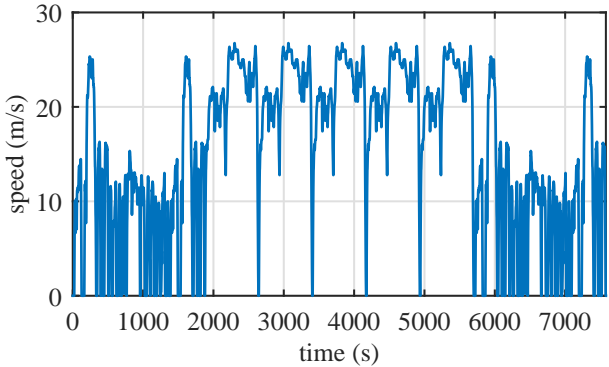}
\caption{\startmodif Combined Driving Cycle. \stopmodif}
\label{fig:CombinedCycle} 
\end{figure}


\section{Total Cost Minimization}
\label{sec:Section3}

The proposed approach is to minimize the overall cost given by the sum of three items: the grid energy for battery recharge, the damage to the battery and the fuel consumed to generate power.
Since these quantities are heterogeneous, the cost function is defined as the sum of the related monetary costs and it is referred to as the \textit{driving cost} hereafter.

In particular, in this section, the new formulation is introduced and the features of its optimal solution are discussed. A real-time control policy will be instead object of the next section.

\subsection{Mathematical Modeling}
In a series HEV, the following electrical power balance holds
\begin{equation}
P_{m}(t) = P_{r}(t) + P_{b}(t),
\label{eq:Link}
\end{equation}
where $P_m$ is the traction motor power, $P_r$ is the generated power and $P_b$ is the battery power, which can be computed as
\begin{equation}
P_{b} = v_{b} i_{b}.
\label{eq:BatteryPower}
\end{equation}

From now on, positive current values will correspond to battery discharging. A circuit model of the battery is needed to relate the voltage $v_b$ to the current $i_b$. The voltage is considered to be the sum of an open circuit term and of an ohmic term, accounting for Joule losses
\begin{equation}
v_{b} = v_{oc} - R_{b} i_{b},
\label{eq:BatteryCircuit}
\end{equation}
where the internal resistance $R_b$ is considered to be constant and the open circuit voltage $v_{oc}$ to be an affine function of the battery state of charge  \cite{guzzella_vehicle_2007}
\begin{equation}
v_{oc}(t) = A_{b} q_b(t) + B_{b},
\label{eq:BatteryOCVAffine}
\end{equation}
where $A_{b}$ and $B_{b}$ are parameters fitting the real open circuit voltage of the battery.
In Figures \ref{fig:battery_ocv} and \ref{fig:battery_r0} the open circuit voltage and the internal resistance of a Li-ion cell, intended for use on a pHEV and measured at different SoC levels, are illustrated. The measured values are compared with the proposed models, as defined in the rest of the section.

\begin{figure}
\centering
\includegraphics[]{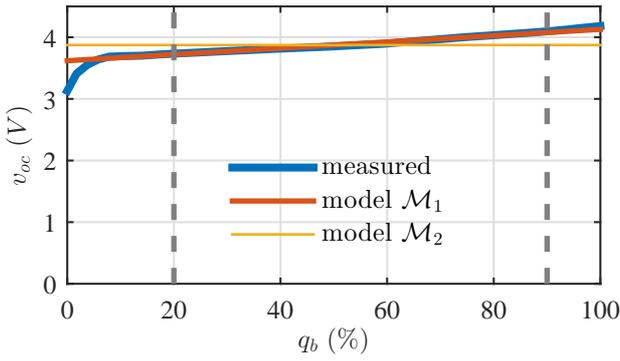}
\caption{Open circuit voltage of a Li-ion cell: measured data vs. proposed models. The accuracy of \eqref{eq:BatteryOCVAffine} and \eqref{eq:BatteryOCVConstant} in describing $v_{oc}$ can here be appreciated.}
\label{fig:battery_ocv}
\end{figure}

The state of charge of the battery $q_b \in \left[0,1\right]$ is defined according to the well established definition \cite{guzzella_vehicle_2007}
\begin{equation}
\dot{q}_{b}(t) = - \frac{i_{b}(t)}{Q_{b}},
\label{eq:BatterySoc}
\end{equation}
\startmodif
where $Q_b$ is the battery capacity, which slowly decreases as the battery ages; generally, the battery is considered \emph{dead} when its capacity has decreased to 80\% of the nominal capacity \cite{Serrao2011}.
Since this decay is very slow, $Q_b$ can be considered constant over a driving cycle of few hours, and can be modeled as
\begin{equation}
Q_{b} = Q_{b}^{nom} (1 - 0.2 \bar{\xi}_{b}),
\label{eq:BatteryCapacityConstant}
\end{equation}
where $Q_{b}^{nom}$ is the battery nominal capacity at the beginning of life and $\bar{\xi}_{b} \in \left[0,1\right]$ is the battery state of health at the beginning of the mission.
\stopmodif
The battery depth of discharge is also introduced as
\begin{equation}
d_{b}(t) = 1 - q_{b}(t).
\label{eq:BatteryDoD}
\end{equation}

\begin{figure}
\centering
\includegraphics[]{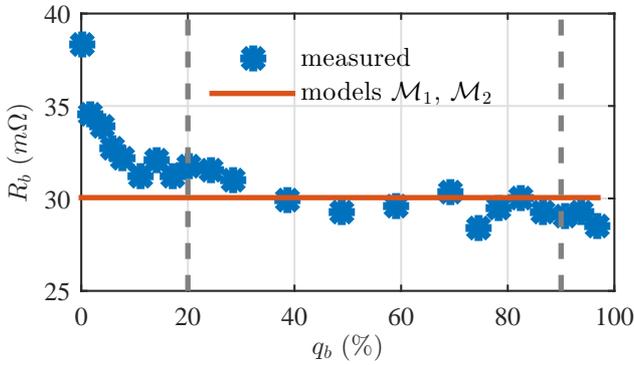}
\caption{Internal resistance of a Li-ion cell: measured data vs. proposed models.}
\label{fig:battery_r0}
\end{figure}

Battery aging is a rather complex phenomenon to model. It is widely discussed in the literature and still lacks a unified approach \cite{ramadesigan2012modeling}. Quite intuitively, aging models are developed for different purposes: manufacturers may be interested in models that accurately describe the underlying electro-chemical processes, and thus can reduce the need for expensive experimental campaigns -- in general, at the cost of high mathematical and computational complexity. On the other hand, simpler aging models are sought when it comes to estimate the battery aging from real world measurements, either to assess the current state of health \cite{jan_experimentally_2012}, or to develop high level strategies \cite{shiau2010optimal}, as it is in our case.

Therefore, the simple but effective aging model \cite{marano2009lithium} is employed
\begin{equation}
\dot{\xi}_{b}(t) = \frac{ \sigma_{b} }{N_b Q_{b}^{nom}} \lvert i_{b}(t) \rvert,
\label{eq:BatterySoh}
\end{equation}
\startmodif
meaning that the depletion of battery life is a function of the current throughput.
\stopmodif
$N_b$ is the number of charge-discharge cycles that the battery can stand over its entire life, in nominal conditions. $\sigma_b$ is a weighting factor, often called \textit{severity factor}, that depends on the operating condition \cite{marano2009lithium}; in our work $\sigma_b$ depends on the battery state of charge and current
\begin{equation}
\sigma_{b} = \sigma_{b}(q_{b}(t),i_{b}(t)).
\label{eq:BatterySigmaVariable}
\end{equation}

Despite its simplicity, this model accounts for battery degradation in a tractable way. A similar model was used in \cite{moura2013battery}, where complex electro-chemical models were reduced to be used in a supervisory controller, coming up with the integral of a static function of current and state of charge and the integral of the current throughput.
A similar approach is also used in \cite{Serrao2011}.
As mentioned in the introduction, \cite{Ebbesen2012} also considers battery aging in the energy management of an HEV. In this case, a different aging model is used, where instead of the severity factor, an explicit function of the battery current is used; however, the aging effects due to the state of charge are not considered with that approach.

The thermal generation unit in a series HEV is mechanically unconstrained from the external world; therefore the mechanical operating point can be arbitrarily chosen as a function of the requested amount of power to generate.
Hereafter it is assumed that a lower level controller continuously adjusts the operating point in a quasi-static manner.
Given the quasi-static efficiency maps of the thermal engine and of the electric generator, the most efficient operating points were computed as a function of the generated electrical power.
If a thermal characterization of the unit was available, this lower level policy could also take into account the effects of thermal transients on fuel economy \cite{arsie2010effects}.
Therefore, under the assumption of quasi-static operation \cite{guzzella_vehicle_2007}, the fuel power and the corresponding flow rate are
\begin{equation}
\begin{split}
P_f(t) &=  \frac{P_{r}(t)}{ \eta_{r}(P_{r}(t))} \delta \left( P_r(t) \right)\\
\dot{m}_{f}(t) &= \frac{P_{f}(t)}{\lambda_{r}}\\
\end{split}
\label{eq:EngineFuelRateMap}
\end{equation}
where $\eta_r$ is the combined efficiency of the Engine-Generator Unit, $\lambda_r$ is the fuel lower heating value and $\delta$ is a Dirac delta.
The above equations imply that the fuel injection is active only when the request of power generation $P_r$ is strictly positive.
The combined efficiency $\eta_r$ is the ratio between generated power $P_r$ and fuel power $P_f$ and is therefore computed from the effciency maps of both the engine and the generator; Figure \ref{fig:egu_efficiency} shows the combined efficiency $\eta_r$ of the Engine-Generator Unit at different power levels.

\begin{figure}
\centering
\includegraphics[]{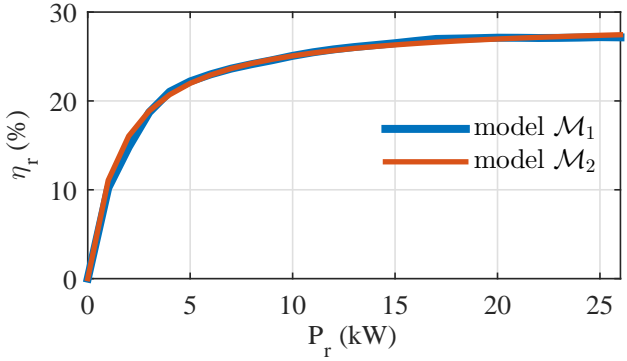}
\caption{ \startmodif Efficiency of the Engine-Generator Unit at different power levels (data taken from the public database ADVISOR \cite{Advisor1999}). \stopmodif}
\label{fig:egu_efficiency}
\end{figure}

The powertrain of the series HEVs is considered to be fully described by the above equations and therefore the following definition is used hereafter.

\begin{definition}
The \textit{full control-oriented model} $\mathcal{M}_1$ of the powertrain of a HEV is defined by the set of Equations~\eqref{eq:Link}-\eqref{eq:EngineFuelRateMap}.
\label{def:FullModel}
\end{definition}

Model $\mathcal{M}_1$ is used to compute the optimal solution by DP in the next subsection, as well as to derive a numeric approximation of the optimal law in Section \ref{sec:Section4}.
By contrast, an explicit formulation of the optimal law is also derived in Section \ref{sec:Section4}, based on a simpler model, with the following differences:
\begin{itemize}
\item The battery open circuit voltage is constant and equal to the nominal voltage, replacing (\ref{eq:BatteryOCVAffine}) with
\begin{equation}
v_{oc}(t) = v_{b}^{nom},
\label{eq:BatteryOCVConstant}
\end{equation}
which is an approximation often used for HEVs and is reasonable if the battery state of charge is kept in a sufficiently narrow range (see again Figure \ref{fig:battery_ocv}).
\item The dependence of the battery severity factor from the state of charge and current is neglected, replacing (\ref{eq:BatterySigmaVariable}) with
\begin{equation}
\sigma_{b} = \bar{\sigma}_{b}(\bar{q}_{b},\bar{i}_{b}),
\label{eq:BatterySigmaConstant}
\end{equation}
where $\bar{q}_{b},\bar{i}_{b}$ represent a nominal operating condition;
\startmodif
this approximation is reasonable if the normal operating ranges of $q_b$ and $i_b$ are sufficiently close to the nominal operating conditions prescribed by the manufacturer.
For instance in \cite{shiau2010optimal,Peterson2010} models of this kind are used and validated for plug-in HEVs; these vehicles (like series HEVs considered here) typically have a quite large $Q_b$, which makes it less likely both to hit the $q_b$ limits and to operate at high C-rate, and thus make a simplified model with constant severity factor reasonable.
\stopmodif
\item The thermal unit is modeled as an affine function of the generated power, replacing (\ref{eq:EngineFuelRateMap}) with
\begin{equation}
\begin{split}
P_f(t) &=  \left( A_r P_{r}(t) + B_r \right) \delta \left( P_r(t) \right)\\
\dot{m}_{f}(t) &= \frac{P_{f}(t)}{\lambda_{r}}\\
\end{split}
\label{eq:EngineFuelRateAffine}
\end{equation}
where $B_r/\lambda_r$ represents the fuel consumption of the idling engine and $A_r$ could be interpreted as the inverse of the generation efficiency. Usually, the efficiency is a function of the generated power, hence the coefficients may be found through a linear fit of the nonlinear model given in \eqref{eq:EngineFuelRateMap}. Moreover, the Willans approach \cite{guzzella_vehicle_2007} models the engine fuel consumption as a linear function of torque; hence, the Willans model corresponds to \eqref{eq:EngineFuelRateAffine}, under the mild assumption that the range engine regime is kept narrow to maximize the efficiency.
\end{itemize}

The simple model is then formally described as follows.

\begin{definition}
The \textit{simplified control-oriented model} $\mathcal{M}_2$ of the powertrain is defined by the set of Equations~\eqref{eq:Link},~\eqref{eq:BatteryCircuit},~\eqref{eq:BatterySoc},~\eqref{eq:BatteryCapacityConstant},\eqref{eq:BatteryDoD},~\eqref{eq:BatterySoh},~\eqref{eq:BatteryOCVConstant},~\eqref{eq:BatterySigmaConstant},
~\eqref{eq:EngineFuelRateAffine}.
\label{def:SimpleModel}
\end{definition}

A formal statement of the problem defined in the previous section is now given.
Let the battery state of charge be the state variable
\begin{equation}
x(t) = q_b(t),
\label{eq:StateChoice}
\end{equation}
the battery current be the input variable
\begin{equation}
u(t) = i_b(t),
\label{eq:InputChoice}
\end{equation}
and the traction motor power be the exogenous disturbance
\begin{equation}
w(t) = P_m(t).
\label{eq:DisturbanceChoice}
\end{equation}

The state function is defined as
\begin{equation}
f(x(t),u(t)) = \dot{q}_b(t),
\label{eq:TCMSStateFunction}
\end{equation}
and the running cost as
\begin{equation}
g(x(t),u(t)) = \alpha v_{b}^{nom} Q_{b} \dot{d}_{b}(t) + \beta v_{b}^{nom} Q_{b}^{nom} \dot{\xi}_{b}(t) + \gamma P_{f},
\label{eq:TCMSIntegralCost}
\end{equation}
which sums three cost items: the grid energy used to recharge the battery, the damage caused to the battery, the fuel consumed to generate power. 
$\alpha, \beta, \gamma$ are respectively the monetary costs of $1 Wh$ of grid energy, $1 Wh$ of battery capacity and $1 Wh$ of fuel energy.
In the equation above, we also included the battery nominal voltage $v_b^{nom}$ and the grid recharging efficiency $\eta_{grid}$.
The first term therefore accounts for the cost of a grid charge \textit{before or after} the trip at hand; the charging phase is thus described by a static model with a constant battery voltage $v_{b}^{nom}$.

Notice that the definitions of $\dot{\xi}_b$ and $\dot{m}_f$ are different for Model $\mathcal{M}_1$ and Model $\mathcal{M}_2$.
Notice also that, although the models of battery aging and of fuel consumption are dynamic models, $\xi_b$ and $m_f$ are not treated as state variables in the control problem statement. This is usually done with $m_f$, but the same treatment can be extended also to $\xi_b$.
As a matter of fact, this is consistent with the approximation of considering $Q_b$ constant as in~\eqref{eq:BatteryCapacityConstant}: the advantage of considering its dependence on $\xi_b$ would be negligible when considering a driving cycle of few hours.
This is also shown in a simulation example in Subsection 3.3.
\emph{By contrast, it is worth considering all the resulting contributions of $\dot{d}_b$, $\dot{\xi}_b$ and $\dot{m}_f$ in the cost function} since, for reasonable scenarios, they share the same order of magnitude.
Since the three quantities are heterogeneous, their monetary costs can be simply summed up, considering the unitary costs $\alpha, \beta, \gamma$.

Considering the constraints,
\begin{equation}
X = [\bar{q}_{min} , \bar{q}_{max}]=[0.2 , 0.9],
\label{eq:TCMSStateSet}
\end{equation}
represent the static bounds for the state variable, while no penalization function on the final state is set, that is
\begin{equation}
h(x(T)) = 0.
\label{eq:TCMSFinalCost}
\end{equation}
Moreover, since the set of admissible final states is not restricted,
\begin{equation}
X_T = X.
\label{eq:TCMSFinalStateSet}
\end{equation}

The time-varying bounds for the control variable $i_b(t)$ are instead
\begin{equation}
U = [u_{min}(t),u_{max}(t)],
\label{eq:TCMSInputSet}
\end{equation}
computed from the corresponding power bounds according to the battery circuit equation
\begin{align}
u_{min}(t) &= \frac{v_{oc} - \sqrt{v_{oc}^2 - 4 R_b P_{min}(t)}}{2 R_b}, \label{eq:TCMSInputMin} \\
u_{max}(t) &= \frac{v_{oc} - \sqrt{v_{oc}^2 - 4 R_b P_{max}(t)}}{2 R_b}. \label{eq:TCMSInputMax}
\end{align}
The time-varying power bounds $P_{min},P_{max}$ ensure that static power bounds of the battery and of the generator are fulfilled with the current value of the traction motor power. In view of \eqref{eq:Link}, they can be computed as:
\begin{align}
P_{min}(t) &= \max \left\lbrace \bar{P}_b^{min} , P_m(t) - \bar{P}_r^{max} \right\rbrace ,\\
P_{max}(t) &= \min \left\lbrace \bar{P}_b^{max} , P_m(t) - \bar{P}_r^{min} , \frac{v_{oc}^2}{4 R_b} \right\rbrace,
\label{eq:TCMSPowerBounds}
\end{align}
where $\bar{P}_b^{min},\bar{P}_b^{max}$ are the battery power limits and $\bar{P}_r^{min},\bar{P}_r^{max}$ are the Engine-Generator unit power limits.
All the model parameters and bounds defined so far are listed in Table \ref{tab:ModelParams}.

Then, the problem of \textit{minimizing the total driving cost over the given driving cycle}, which will be called from now on \textit{Total Cost Minimization Strategy} (TCMS), can be mathematically stated as follows:
\begin{equation}
\label{eq:ourOCP}
\begin{aligned}
& \underset{u}{\text{min}}
& & \int_0^T \left(\alpha v_{b}^{nom} Q_{b} \dot{d}_{b}(t) + \beta v_{b}^{nom} Q_{b}^{nom} \dot{\xi}_{b}(t) + \gamma P_{f}(t) \right) dt \\
& \text{s. t.}
& & \dot{q}_{b}(t) = - \frac{i_{b}(t)}{Q_{b}}\\
&
& & q_b(0) = q_0\\
&
& & q_b(t) \in [0.2,0.9]\\
&
& & i_b(t) \in [u_{min}(t),u_{max}(t)].
\end{aligned}
\end{equation}

\subsection{Benchmark Optimal Solution}
The solution to Problem~\eqref{eq:ourOCP} can be computed off-line based on Dynamic Programming \cite{sniedovich2010dynamic}. Since such an approach does not directly apply to continuous time systems, a discrete-time counterpart of the state function is considered.
By Backward Euler approach, (\ref{eq:BatterySoc}) yields
\begin{equation}
\label{eq:DiscreteQb}
q_{b}(k+1) = q_{b}(k) - \frac{T_s i_{b}(k)}{Q_{b}}, \; k = 0,...,N-1,
\end{equation}
$T_s$ being the sampling time and $T=N T_s$ being the time horizon of the mission.

The cost function $J$ to minimize is redefined in the discrete time domain as
\begin{equation}
h(x(N)) + \sum_{k=0}^{N-1} g(x(k),u(k),w(k)),
\end{equation}
where
\begin{multline}
g(x(k),u(k),w(k)) = \alpha v_{b}^{nom} Q_{b} \frac{d_{b}(k+1)-d_{b}(k)}{T_{s}}\\
+ \beta v_{b}^{nom} Q_{b}^{nom} \frac{\xi_{b}(k+1)-\xi_{b}(k)}{T_{s}} + \gamma \frac{P_{f}(k+1)-P_{f}(k)}{T_{s}}.
\end{multline}
If Backward Euler approach is used, \eqref{eq:BatterySoh},\eqref{eq:EngineFuelRateMap} yield
\begin{gather}
\xi_{b}(k+1) = \xi_{b}(k) + T_s \frac{ \sigma_{b}(k) \lvert i_{b}(k) \rvert }{N_b Q_{b}^{nom}} \label{eq:DiscreteXib}\\
P_{f}(k+1) = P_{f}(k) + T_s \frac{P_{r}(k)}{ \eta_{r}(P_{r}(k)) \delta(P_r(k))} \label{eq:DiscretePf}
\end{gather}

According to Dynamic Programming, the optimal cost for a given initial state $J^*(x(0))$ is found at the last iteration of the following algorithm, proceeding backwards in time from $k=N-1$ to $k=0$
\begin{equation}
\begin{aligned}
J_N(x(N)) &= h(x(N))\\
J_{k}(x(k)) &= \min_u \lbrace g(x(k),u(k),w(k)) \\
&+ J_{k+1}(f(x(k),u(k),w(k)) \rbrace .
\end{aligned}
\end{equation}
The optimal control policy $\mathcal{\pi}_{DP} = \lbrace \mu_0^*, ... \mu_{N-1}^* \rbrace$ is then found as
\begin{multline}
u^*(k) = \mu_{k}^{*}(x(k)) = \arg \min_u \lbrace g(x(k),u(k),w(k)) \\ 
+ J_{k+1}(f(x(k),u(k),w(k)) \rbrace  , \forall k , \forall x(k).
\end{multline}

Consider now the simulation environment defined in Section \ref{sec:simulator}.
For the Urban Driving Cycle, the values of the optimal cost-to-go function $J_k(x(k))$ are depicted in Fig.~\ref{fig:JmapUrban}. Intuitively, $J_k(x(k))$ mainly grows when moving backwards in time (\textit{i.e.} in distance traveled); the dependence on the state of charge is minor, because no penalty for the final state of charge is set and the considered trip is within the electric range.

\begin{figure}
\centering
\psfragfig{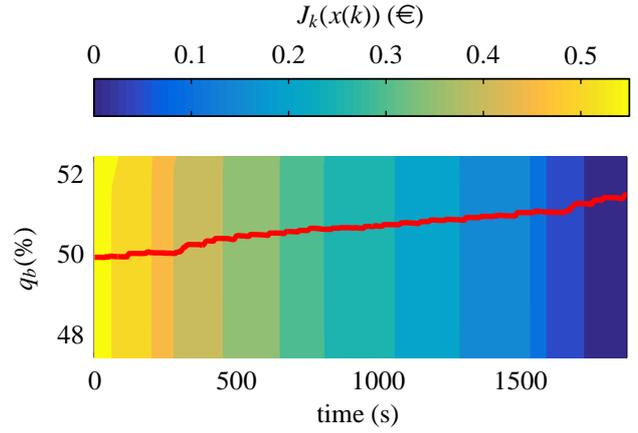}
\caption{\startmodif Optimal Cost-to-go function $J_k(x(k))$ with the TCMS in the Urban Driving Cycle. \stopmodif}
\label{fig:JmapUrban}
\end{figure}

The optimal state trajectory, given an initial state of charge $q_b(0)=0.5$, is highlighted with a red line.
\startmodif
Fig.~\ref{fig:UmapUrban} shows a portion of the corresponding engine operation.
As the figure shows, the engine outputs power during the vehicle accelerations, while it is switched off when the vehicle speed (i.e. the demanded $P_m$) starts decreasing.
\stopmodif

\begin{figure}
\centering
\includegraphics[]{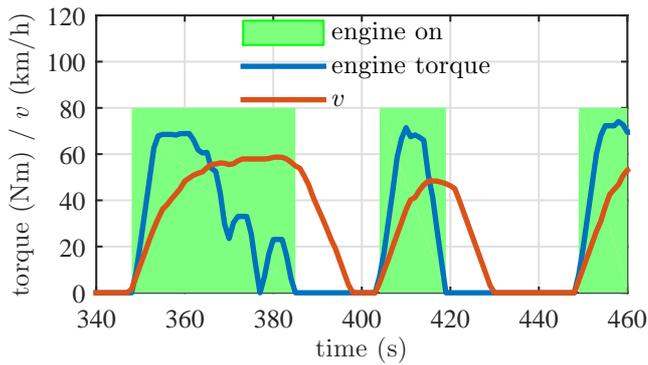}
\caption{\startmodif Engine torque and vehicle speed $v$ with the TCMS in a portion of the Urban Driving Cycle. The shaded areas indicate where the engine is on. \stopmodif}
\label{fig:UmapUrban}
\end{figure}

The TCMS approach is compared to the approaches presented in Section~\ref{sec:Section2}.
\startmodif
Analogously to the TCMS formulation, also the other approaches are implemented using Dynamic Programming; details on the formulation are in \ref{sec:AppendixC}.
As noted above, the simulation scenario encompasses a CNG range-extender, which makes the trade-off between thermal and electric power non trivial.
\stopmodif
The trends of the battery state of charge under the Urban Driving Cycle are compared in Fig.~\ref{fig:BenchmarkUrban}.
The corresponding final values of the total driving cost and fuel consumption are given in Table~\ref{tab:BenchmarkUrban}.

The Full Electric strategy depletes more than 5\% of the charge, the Charge Sustaining strategy and the ECMS attain the same final state of charge, while the TCMS performs a slight recharge of about 2\%.
Notice that the ECMS was tuned on this driving cycle to obtain perfect charge sustenance.

The TCMS attains the minimum driving cost, the Full Electric strategy is nearly 50\% more expensive while the Charge Sustaining strategy and the ECMS have about 30\% higher cost.
Despite some deviations in the trend of the state of charge, the Charge Sustaining strategy and the ECMS attain the same cost and almost the same fuel consumption.
On the other hand, the TCMS uses the most fuel, while the Full Electric simply keeps the engine off for all the cycle.
\startmodif
Notice that different combinations of the cost coefficients yield a different balance between thermal and electrical power; for instance, it is intuitive that, if $\alpha=\beta=0$, the TCMS behaves as the Full Electric strategy.
For any combination, the TCMS attains at least the same performance (in terms of monetary cost) as one of the other approaches.
The interested reader can find in Section~\ref{sec:Section5} a sensitivity analysis of the performance of the proposed approach for a broad range of fuel prices.
\stopmodif

\begin{figure}
\centering
\includegraphics[]{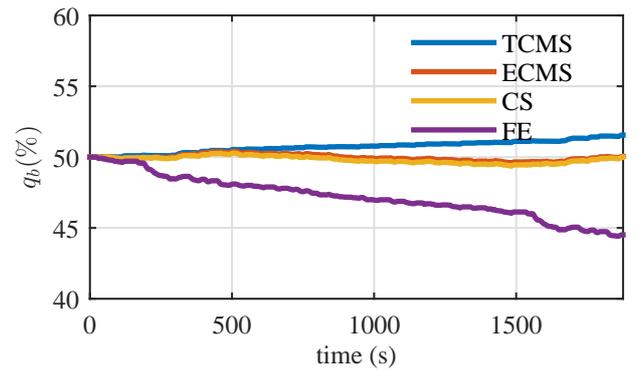}
\caption{Battery state of charge in the Urban Driving Cycle according to the considered approaches: TCMS, Charge Sustaining, ECMS, Full Electric.}
\label{fig:BenchmarkUrban}
\end{figure}

\begin{table}
\centering
\caption{Driving cost and fuel consumption in the Urban Driving Cycle according to the considered approaches: TCMS, Charge Sustaining, ECMS, Full Electric.}
\begin{tabular}{ccccc}
\toprule
 & TCMS & ECMS & CS & FE \\
\midrule
Driving Cost $[\text{\euro}]$ 	& 0.57 		& 0.74 		& 0.74 		& 0.82 \\
Fuel Mass $[g]$ 				& 484.29 	& 343.56 	& 340.65 	& 1.49 \\
\bottomrule
\end{tabular}
\label{tab:BenchmarkUrban}
\end{table}

The same comparison is carried out in Fig.~\ref{fig:BenchmarkCombined} and Table~\ref{tab:BenchmarkCombined}, in the case of the Combined Driving Cycle. The main observations are as follows.

First of all, since such a cycle is beyond the electric range of the vehicle, also the Full Electric strategy needs the thermal unit to complete the trip; nonetheless, it completely depletes the battery charge.
In the Combined Driving Cycle, the TCMS recharges the battery by about 5\%. Some features are analogous to the Urban Driving Cycle: the Charge Sustaining strategy and the ECMS attain the same final state of charge, although with different transients; the TCMS is the cheapest strategy; the Charge Sustaining strategy and the ECMS attain almost the same results and Full Electric strategy is the most expensive.
In terms of fuel consumption, the TCMS uses the most fuel, the Charge Sustaining strategy and the ECMS have similar results and the Full Electric strategy attains the lowest consumption.

\begin{figure}
\centering
\includegraphics[]{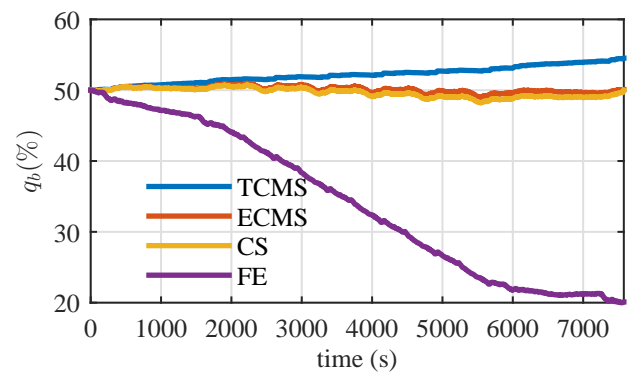}
\caption{Battery state of charge in the Combined Driving Cycle according to the considered approaches: TCMS, Charge Sustaining, ECMS, Full Electric.}
\label{fig:BenchmarkCombined}
\end{figure}

\begin{table}
\centering
\caption{Driving cost and fuel consumption in the Combined Driving Cycle according to the considered approaches: TCMS, Charge Sustaining, ECMS, Full Electric.}
\begin{tabular}{ccccc}
\toprule
 & TCMS & ECMS & CS & FE \\
\midrule
Driving Cost $[\text{\euro}]$ & 3.07 & 4.24 & 4.25 & 4.59 \\
Fuel Mass $[g]$ & 2723.4 & 2196.8 & 2192.8 & 520.0 \\
\bottomrule
\end{tabular}
\label{tab:BenchmarkCombined}
\end{table}

From the above results, a couple of interesting facts can be inferred:
\begin{itemize}
	\item it is a common belief that, when a driving cycle is fully achievable in Full electric mode, this is the least costly policy, according to the assumption that ``electric is cheaper". This is not always true, and in particular it is not in the considered \textit{scenario}.
	\item the optimal usage of the ICE, in terms of monetary cost, is not trivial.
\end{itemize}

As a final remark, it is reasonable to expect that in the near future -- thanks to technological advances -- range extender will be a more flexible component of the powertrain; notable examples are multi-fuel engines and plug-in range extenders. As a future task, the proposed approach could quite easily be extended to account for cases such as (i) only expensive fuel is available in the first part of the trip whereas (after a gas station is reached) it becomes cheaper in the second part, and (ii) the range extender is only available in a portion of the trip. In \textit{scenarios} like the ones above, the proposed formulation in terms of monetary cost looks more appropriate and natural than more traditional formulations in terms of energy consumption.

\subsection{State variable choice}

In this subsection the choice of using only one state variable ($q_b$) instead of two ($q_b$ and $\xi_b$) is motivated by means of a simulation.
We show that in the scenario considered above, the effect of considering $\xi_b$ as a state variable is negligible.
This is accomplished modifying the problem formulation for the implementation of the DP:
\begin{itemize}
\item $\xi_b$ is considered as a second state variable
\item the capacity in equation \eqref{eq:BatterySoc} is computed as $Q_{b} = Q_{b}^{nom} (1 - 0.2 \xi_{b}(t))$
\end{itemize}

The results found with this formulation are compared with the proposed formulation in Figure \ref{fig:AgingState} and in Table \ref{tab:AgingState}: no significant difference is found neither in the state of charge trend, nor in the overall cost and fuel mass.

\begin{figure}
\centering
\includegraphics[]{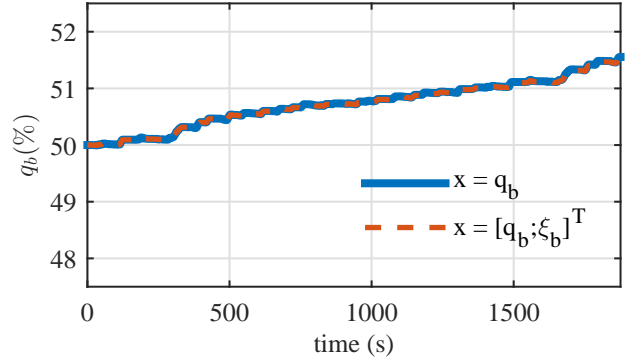}
\caption{Battery state of charge in the Urban Driving Cycle with the TCMS approach: DP implementation with 1 and 2 state variables.}
\label{fig:AgingState}
\end{figure}

\begin{table}
\centering
\caption{Driving cost and fuel consumption in the Urban Driving Cycle with the TCMS approach: DP implementation with 1 and 2 state variables.}
\begin{tabular}{ccc}
\toprule
 & $x=q_b$ & $x=[q_b;\xi_b]^T$ \\
\midrule
Driving Cost $[\text{\euro}]$ & 0.57 & 0.57 \\
Fuel Mass $[g]$ & 484.29 & 483.24 \\
\bottomrule
\end{tabular}
\label{tab:AgingState}
\end{table}


\section{Real-Time Optimal Control}
\label{sec:Section4}

The results obtained in the previous section rely on the \textit{a-priori} knowledge of the driving cycle. Therefore, the optimal trajectory of $i_b(t)$ cannot be implemented in real-time.

In this section, an analytic and a numerical causal control policy solving Problem~\eqref{eq:ourOCP} are presented. 
By relying on the Pontryagin's Minimum Principle \cite{bertsekas1995dynamic}, performance close to the optimum are attained also in real-time, without knowledge of the driving cycle. At the end of this section, such policies are compared with the benchmark results presented in Section~\ref{sec:Section3}.

\subsection{Unconstrained Explicit Optimal Control Law}

A result valid for the Model $\mathcal{M}_2$ is first given discarding the constraints $u_{min}, u_{max}$ on the control variable.
In the statement we refer to the adjoint state for the optimal control problem: a discussion on this variable is given in Subsection 4.4.

\begin{thm}
Consider Model $\mathcal{M}_2$ given in Definition~\ref{def:SimpleModel}.
The optimal battery current ${i}_b^o$ is a function of motor power $P_{mot}$ and of the adjoint state $p$
\begin{equation}
{i}_b^o = 
\left\lbrace
\begin{split}
{i}_b^{o,1} &= \left( \frac{p}{Q_b} - \alpha v_{oc} + \beta \frac{\sigma_b v_{oc}}{N_b} + \gamma A_r v_{oc} \right) \frac{1}{2 A_r R_b \gamma} \\
&\qquad if \, P_m > 0 \cap P_m > P_{lim}^{(1,4)}(p) \cap p<p_{lim}^{(1,2)} \\
{i}_b^{o,2} &= 0 \\
&\qquad if \, P_m > 0 \cap P_m > P_{lim}^{(2,4)}(p) \cap p_{lim}^{(1,2)}<p<p_{lim}^{(2,3)} \\
{i}_b^{o,3} &= \left( \frac{p}{Q_b} - \alpha v_{oc} - \beta \frac{\sigma_b v_{oc}}{N_b} + \gamma A_r v_{oc} \right) \frac{1}{2 A_r R_b \gamma} \\
&\qquad if \, P_m > 0 \cap P_m > P_{lim}^{(3,4)}(p) \cap p>p_{lim}^{(2,3)} \\
{i}_b^{o,1} &= \left( \frac{p}{Q_b} - \alpha v_{oc} + \beta \frac{\sigma_b v_{oc}}{N_b} + \gamma A_r v_{oc} \right) \frac{1}{2 A_r R_b \gamma} \\
&\qquad if \, P_m < 0 \cap P_m > P_{lim}^{(1,5)}(p) \cap p<p_{lim}^{(1,2)} \\
{i}_b^{o,4} &= {i}_b^{o,5} = \left( v_{oc} - \sqrt{ v_{oc}^2 - 4 R_b P_{m}} \right) \frac{1}{2 R_b} \\
&\qquad else \\
\end{split}
\right.
\label{eq:OptimalCurrent}
\end{equation}
where
\begin{equation}
\begin{split}
p_{lim}^{(1,2)} &= Q_b \left( \alpha v_{oc} - A_r \gamma v_{oc} - \frac{\sigma_b \beta v_{oc}}{N_b} \right) \\
p_{lim}^{(2,3)} &= Q_b \left( \alpha v_{oc} - A_r \gamma v_{oc} + \frac{\sigma_b \beta v_{oc}}{N_b} \right) \\
\end{split}
\end{equation}
and the power limits separating the different regions are given in Equation \eqref{eq:PowerLimits} in Appendix B.
\end{thm}

\begin{proof}
The proof refers to Pontryagin's Minimum Principle \cite{bertsekas1995dynamic}.
Define the Hamiltonian function as
\begin{equation}
H(q_b,i_b,p) = g(q_b,i_b) + p^T f(q_b,i_b),
\end{equation}
where $p$ is a dynamic variable, often referred to as \textit{adjoint state}, obeying to
\begin{equation}
\dot{p}(t) = -\nabla_{q_b} H(q_b^{*}(t),i_b^{*}(t),p(t)),
\end{equation}
and subject to the boundary condition
\begin{equation}\label{eq:boundary}
p(T) = \nabla_{q_b} h(q_b^{*}(T)).
\end{equation}
The Hamiltonian reads
\begin{equation}
H(t) = 
\left\lbrace
\begin{split}
H_1(t) \quad &if \quad P_b \leq 0 \leq P_m - P_b\\
H_2(t) \quad &if \quad P_b = 0 \leq P_m\\
H_3(t) \quad &if \quad 0 \leq P_b \leq P_m\\
H_4(t) \quad &if \quad P_b = P_m \geq 0 \\
H_5(t) \quad &if \quad P_b = P_m \leq 0\\
\end{split}
\right.
\end{equation}
where
\begin{equation}
\begin{split}
H_1 &= \alpha v_{oc} i_{b} - \frac{\beta \sigma_b v_{oc} i_{b}}{N_b} + \gamma \left( A_r \left(P_{m} - i_{b} v_{oc} + R_b i_{b}^2 \right) + B_r \right)  - \frac{p i_{b}}{Q_b}\\
H_2 &= \gamma \left( A_r P_{m} + B_r \right)\\
H_3 &= \alpha v_{oc} i_{b} + \frac{\beta \sigma_b v_{oc} i_{b}}{N_b} + \gamma \left( A_r \left(P_{m} - i_{b} v_{oc} + R_b i_{b}^2 \right) + B_r \right)  - \frac{p i_{b}}{Q_b}\\
H_4 &= \alpha v_{oc} i_{b} + \frac{\beta \sigma_b v_{oc} i_{b}}{N_b} - \frac{p i_{b}}{Q_b} \\
H_5 &= \alpha v_{oc} i_{b} - \frac{\beta \sigma_b v_{oc} i_{b}}{N_b} - \frac{p i_{b}}{Q_b}\\
\end{split}
\end{equation}
By minimizing the Hamiltonian, the optimal control law is found as
\begin{equation}
\begin{split}
i_b^{o,1} &= \left( \frac{p}{Q_b} - \alpha v_{oc} + \beta \frac{\sigma_b v_{oc}}{N_b} + \gamma A_r v_{oc} \right) \frac{1}{2 A_r R_b \gamma} \\
i_b^{o,2} &= 0 \\
i_b^{o,3} &= \left( \frac{p}{Q_b} - \alpha v_{oc} - \beta \frac{\sigma_b v_{oc}}{N_b} + \gamma A_r v_{oc} \right) \frac{1}{2 A_r R_b \gamma} \\
i_b^{o,4} &= i_b^{o,5} = \left( v_{oc} - \sqrt{ v_{oc}^2 - 4 R_b P_{m}} \right) \frac{1}{2 R_b} \\
\end{split}
\end{equation}
where $i_b^{o,2}, i_b^{o,4} i_b^{o,5}$ come directly from the definition, while $i_b^{o,1}, i_b^{o,3}$ are found analyzing the first and second derivatives of the Hamiltonian
\begin{equation}
\begin{split}
\nabla_u H_1 &= \alpha v_{oc} - \frac{\beta \sigma_b v_{oc}}{N_b} - \gamma A_r \left(v_{oc} - 2 R_b i_{b}\right) - \frac{p}{Q_b}\\
\nabla_u H_3 &= \alpha v_{oc} + \frac{\beta \sigma_b v_{oc}}{N_b} - \gamma A_r \left(v_{oc} - 2 R_b i_{b}\right) - \frac{p}{Q_b}\\
\nabla_u^2 H_1 &= 2 A_r R_b \gamma > 0 \\
\nabla_u^2 H_3 &= 2 A_r R_b \gamma > 0 \\
\end{split}
\end{equation}
The second derivatives are both positive because ${A_r} , {R_b} , {\gamma}$ are positive parameters, therefore, the Hamiltonian is convex.
The boundaries between modes 1, 2 and 3 are found studying the limits of the above gradients when $i_b$ approaches zero
\begin{equation}
\begin{split}
\nabla_u H_1 &\xrightarrow{i_b \rightarrow 0^-} \alpha v_{oc} - \frac{\beta \sigma_b v_{oc}}{N_b} - \gamma A_r v_{oc} - \frac{p}{Q_b}\\
\nabla_u H_3 &\xrightarrow{i_b \rightarrow 0^+} \alpha v_{oc} + \frac{\beta \sigma_b v_{oc}}{N_b} - \gamma A_r v_{oc} - \frac{p}{Q_b}\\
\end{split}
\end{equation}
Since $\beta, \sigma_b, v_{oc}, N_b$ are positive and $H_1, H_3$ are convex, then $\nabla_u H_3 > \nabla_u H_1$.
For the same reason, the minimum among $H_1,H_2,H_3$ is found studying the sign of the two expressions above.
$H_1$ is minimum when $\nabla_u H_1 > 0 \cap \nabla_u H_3 > 0$, i.e.
\begin{equation}
p < p_{lim}^{(1,2)} = Q_b \left( \alpha v_{oc} - A_r \gamma v_{oc} - \frac{\sigma_b \beta v_{oc}}{N_b} \right)
\end{equation}
$H_3$ is minimum when $\nabla_u H_1 < 0 \cap \nabla_u H_3 < 0$, i.e.
\begin{equation}
p > p_{lim}^{(2,3)} =  Q_b \left( \alpha v_{oc} - A_r \gamma v_{oc} + \frac{\sigma_b \beta v_{oc}}{N_b} \right)
\end{equation}
$H_2$ is minimum when $\nabla_u H_1 < 0 \cap \nabla_u H_3 > 0$, i.e.
\begin{equation}
p_{lim}^{(1,2)} < p < p_{lim}^{(2,3)}
\end{equation}

Consider first the case $P_m>0$.
In each of the three regions just defined, the optimal Hamiltonian among $H_1,H_2,H_3$ is compared to $H_4$; the power limits $P_{lim}^{(1,4)}, P_{lim}^{(2,4)}, P_{lim}^{(3,4)}$ are the boundaries of the pure electric mode in each region
\begin{equation}
\begin{split}
p < p_{lim}^{(1,2)} \rightarrow P_{lim}^{(1,4)}=\lbrace P_m | H_1<H_4 \rbrace \\
p_{lim}^{(1,2)} < p < p_{lim}^{(2,3)} \rightarrow P_{lim}^{(2,4)}=\lbrace P_m | H_2<H_4 \rbrace \\
p > p_{lim}^{(2,3)} \rightarrow P_{lim}^{(3,4)}=\lbrace P_m | H_3<H_4 \rbrace \\
\end{split}
\end{equation}

Consider now the case $P_m<0$.
Modes 2 and 3 are not feasible since generated power must be non negative.
Therefore, when $p > p_{lim}^{(1,2)}$ the optimal Hamiltonian is $H_5$, otherwise both $H_1$ and $H_5$ are feasible candidates.
The line dividing the two regions where $H_1$ and $H_5$ are optimal is
\begin{equation}
p < p_{lim}^{(1,2)} \rightarrow  P_{lim}^{(1,5)}=\lbrace P_m | H_1<H_5 \rbrace \\
\end{equation}

\end{proof}

The explicit optimal control law defines five different regions, each having a different expression of the optimal current in terms of the motor power and the adjoint state.
The different regions are shown in Figure \ref{fig:expl_map_unconstr} when considering the powetrains parameters used in our simulations.

\begin{figure}
\centering
\includegraphics[width=0.9\columnwidth]{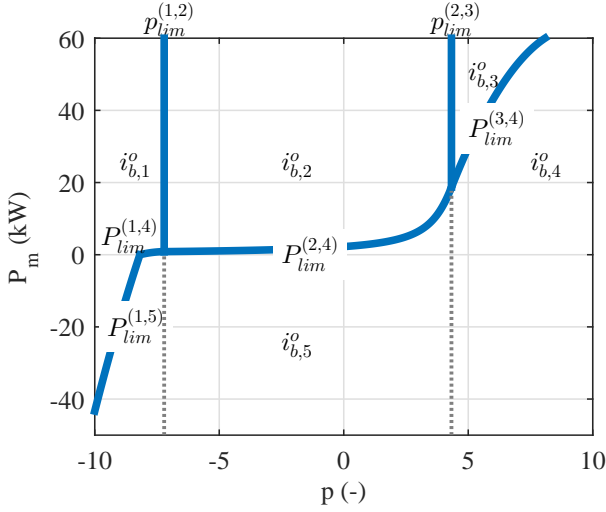}
\caption{Regions of the explicit optimal control policy (unconstrained case).}
\label{fig:expl_map_unconstr}
\end{figure}

\subsection{Constrained Explicit Optimal Control Law}

A result valid for Model $\mathcal{M}_2$ is now given including the constraints $u_{min}, u_{max}$ on the control variable, as defined in Equations \eqref{eq:TCMSInputMin}, \eqref{eq:TCMSInputMax}.
These constraints reflect the bounds on both the battery and the engine power.
The map of the constrained optimal control is depicted in Figure \ref{fig:expl_map_constr}.
The derivation of the map follows directly from the application of the power bounds to the unconstrained map.
Notice that in this case the system has 7 possible modes, instead of 5 in the unconstrained case.
For this reason, the new limits $(P_{b}^{min}+P_{g}^{max}),P_{lim}^{(1,*)},P_{g}^{max},P_{lim}^{(3,*)},P_{b}^{max},P_{lim}^{(1,5 *)}$ arise.
As for the adjoint state, the limits $p_{lim}^{(1)}, p_{lim}^{(3)}$ arise as well.
Notice that, for $p<p_{lim}^{(1)}$, the boundary between thermal and pure electric modes is defined by $P_{lim}^{(1,5 *)}$ which is computed as $P_{lim}^{(1,5 )}$ but considering that in this condition the battery power is saturated $P_b=P_b^{min}$.

In the regions when constraints are active, the saturated currents $i_{min},i_{max}$ are defined as $u_{min},u_{max}$ given in \eqref{eq:TCMSInputMin}, \eqref{eq:TCMSInputMax}. Where the battery power is saturated to $\bar{P}_b^{min}$ it is $i_{min}(\bar{P}_b^{min})=u_{min}(\bar{P}_b^{min})$, while the maximum power bound $\bar{P}_b^{max}$ is reached, it becomes $i_{max}(\bar{P}_b^{max})=u_{max}(\bar{P}_b^{max})$. Instead, when the generator power is saturated to $\bar{P}_r^{max}$, the optimal current becomes $i_{min}(P_m-\bar{P}_r^{max})=u_{min}(P_m-\bar{P}_r^{max})$.

\begin{figure}
\centering
\includegraphics[width=0.97\columnwidth]{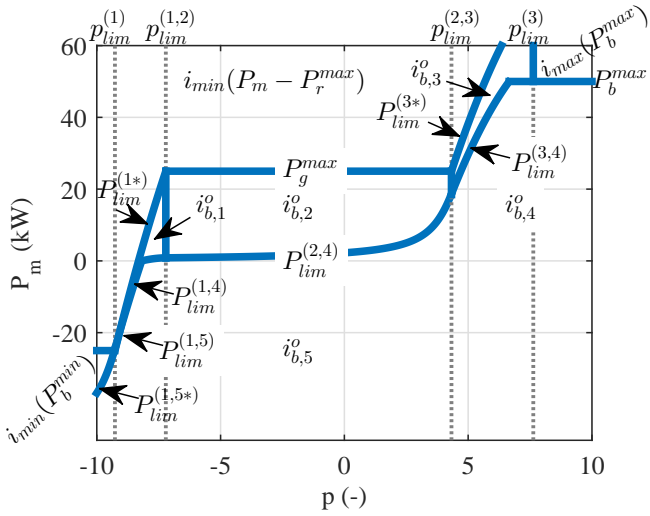}
\caption{Regions of the explicit optimal control policy (constrained case).}
\label{fig:expl_map_constr}
\end{figure}

In the map depicted in Figure \ref{fig:expl_map_constr}, different bounds on the battery and engine power would change the constraints active in each mode, thus changing the shape of the map.

An alternative way of deriving the optimal control has also been proposed for an ECMS policy in \cite{Ambuhl2010}: the Hamiltonian is evaluated for the control candidates $\mathcal{U} = \lbrace i_1^o , i_2^o , i_3^o , i_4^o , i_5^o , i_{min} , i_{max} \rbrace$ and the minimization is performed on these candidates only
\begin{equation}
i_b^o = \arg \min_{u\in\mathcal{U}} \lbrace H \rbrace
\end{equation}
This approach has the advantage of using the results of the explicit law for drastically reducing the number of control candidates, with respect to a standard numeric minimization, like the one presented in the next subsection.
On the other hand, the complex analytic expressions of each region in terms of motor power and adjoint state are not needed.

In the remainder of the paper we indicate with $\mathcal{\pi}_{X}$ the control policy depicted described in this subsection.

\subsection{Constrained Numeric Optimal Control Law}

The results presented in this section so far are referred to Model $\mathcal{M}_2$.
We have shown that in that case, the optimal control law can be expressed as an explicit function of the adjoint state and of the disturbance.
As already mentioned in Section \ref{sec:Section3}, Model $\mathcal{M}_2$ introduces some approximation in the battery open circuit voltage, in the battery severity factor and in the engine efficiency.
In this work, the criticality of these approximations is analyzed in the simulations at the end of this section; in a generic application, it has to be verified based on measured data on the system.

In this subsection, we show how to compute the optimal control law when Model $\mathcal{M}_2$ is not representative enough, and the more accurate Model $\mathcal{M}_1$ has to be used.
In such a case, the optimal current is found by numerical minimization of the Hamiltonian
\begin{equation}
i_b^o = \arg \min_{u\in U} \lbrace H \rbrace
\end{equation}
where $U$ is the feasible input set (considering the constraints $u_{min}, u_{max}$) and the Hamiltonian is 
\begin{equation}
H = \alpha v_b^{nom} i_{b} + \frac{\beta \sigma_b v_b^{nom} \lvert i_{b} \rvert}{N_b} + \gamma P_f(P_m,i_b,q_b) - \frac{p i_{b}}{Q_b}
\label{eq:HamiltonianFull}
\end{equation}
Notice that in this case $v_{oc}$ is an affine function of the state and $\sigma_b$ is a function of the state and the battery current.

The policy derived with this approach is referred to as policy $\mathcal{\pi}_{N}$ hereafter.

\subsection{Adjoint State}

For policy $\mathcal{\pi}_{X}$ (both unconstrained and constrained), i.e. when Model $\mathcal{M}_2$ is considered, the adjoint state is  constant
\begin{equation}
\dot{p} = 0
\end{equation}
while for policy $\mathcal{\pi}_{N}$, i.e. when Model $\mathcal{M}_1$ is considered, the adjoint state is subject to the dynamic equation
\begin{equation}
\begin{split}
\dot{p}(t) 	&= -\nabla_{q_b} g({q_b}(t),i_b(t)) - p(t)' \nabla_{q_b} f(q_b(t),i_b(t))\\
			&= -\frac{\beta v_b^{nom}}{N_b} \lvert i_b(t) \rvert \nabla_{q_b}\sigma_b(t) - \gamma \nabla_{q_b} P_f(t).\\
\end{split}
\label{eq:AdjointStateFull}
\end{equation}

The proof of these statements comes directly from the application of Pontryagin's Minimum Principle.

When the final state constraint is not active, the boundary condition for the adjoint state is simply
\begin{equation}
p(T) = \nabla_{q_b} h = 0
\end{equation}

Hereafter, we take the approximation $\dot{p} \approx 0$ to compute the approximated policy $\mathcal{\pi}_N$ in a causal framework.

\subsection{Simulation Results}
The two proposed real-time implementations $\mathcal{\pi}_X, \mathcal{\pi}_N$ of the TCMS were compared to the non-causal implementation $\mathcal{\pi}_{DP}$ presented in Section~\ref{sec:Section3}.
The trends of the battery state of charge under the Urban Driving Cycle are compared in Fig.~\ref{fig:ControlUrban}.
The corresponding final values of the total driving cost and fuel consumption are given in Table~\ref{tab:ControlUrban}.
Apparently policy $\mathcal{\pi}_N$ approximates with high accuracy the results of the Dynamic Programming.
Also the explicit policy $\mathcal{\pi}_X$ gives very close results, with negligible difference in the final cost and an increase in fuel consumption of about 3\%; this is reflected in a slightly higher (about 0.13\%) state of charge at the end of the cycle.

\begin{figure}
\centering
\includegraphics[]{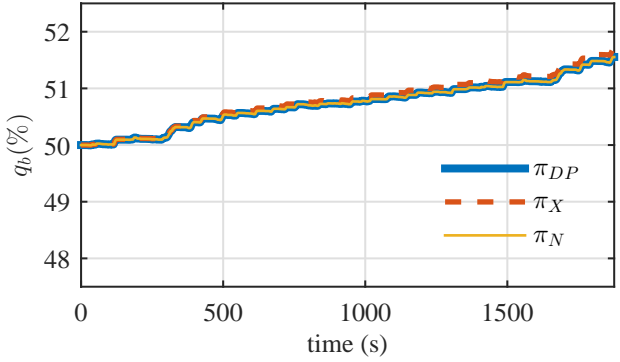}
\caption{Battery state of charge for the Urban Driving Cycle: benchmark solution ($\mathcal{\pi}_{DP}$), analytical real-time solution ($\mathcal{\pi}_{X}$) and numerical real-time solution ($\mathcal{\pi}_{N}$) of Problem \eqref{eq:ourOCP}.}
\label{fig:ControlUrban}
\end{figure}

\begin{table}
\centering
\caption{Driving cost and fuel consumption for the Urban Driving Cycle: benchmark solution ($\mathcal{\pi}_{DP}$), analytical real-time solution ($\mathcal{\pi}_{X}$) and numerical real-time solution ($\mathcal{\pi}_{N}$) of Problem \eqref{eq:ourOCP}.}
\begin{tabular}{cccc}
\toprule
 & $\mathcal{\pi}_{DP}$ & $\mathcal{\pi}_{X}$ & $\mathcal{\pi}_{N}$ \\
\midrule
Driving Cost $[\text{\euro}]$ & 0.566 & 0.568 & 0.566 \\
Fuel Mass $[g]$ & 484.29 & 500.75 & 484.30 \\
\bottomrule
\end{tabular}
\label{tab:ControlUrban}
\end{table}

The same comparison is carried out in Fig.~\ref{fig:ControlCombined} and Table~\ref{tab:ControlCombined}, in the case of the Combined Driving Cycle. 
Similar comments apply: the monetary cost attained by policy $\mathcal{\pi}_X$ is not significantly different from policies $\mathcal{\pi}_N, \mathcal{\pi}_{DP}$; on the other hand, the fuel consumption is about 3.5\% higher.
Also in this case this is reflected in a higher final state of charge: in this case the deviation from the other two policies is more relevant (about 1\%).

\begin{figure}
\centering
\includegraphics[]{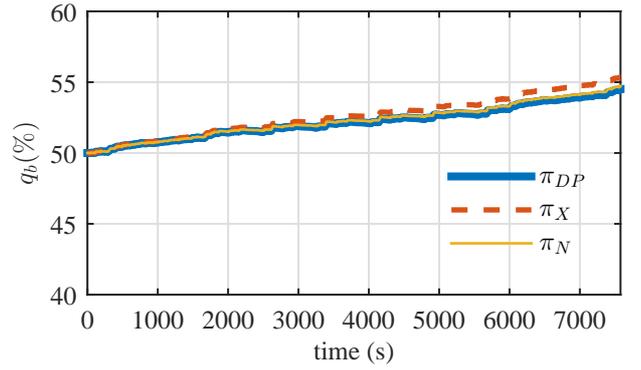}
\caption{Battery state of charge for the Combined Driving Cycle: benchmark solution ($\mathcal{\pi}_{DP}$), analytical real-time solution ($\mathcal{\pi}_{X}$) and numerical real-time solution ($\mathcal{\pi}_{N}$) of Problem \eqref{eq:ourOCP}.}
\label{fig:ControlCombined}
\end{figure}

\begin{table}
\centering
\caption{Driving cost and fuel consumption for the Combined Driving Cycle: benchmark solution ($\mathcal{\pi}_{DP}$), analytical real-time solution ($\mathcal{\pi}_{X}$) and numerical real-time solution ($\mathcal{\pi}_{N}$) of Problem \eqref{eq:ourOCP}.}
\begin{tabular}{cccc}
\toprule
 & $\mathcal{\pi}_{DP}$ & $\mathcal{\pi}_{X}$ & $\mathcal{\pi}_{N}$ \\
\midrule
Driving Cost $[\text{\euro}]$ & 3.070 & 3.073 & 3.070 \\
Fuel Mass $[g]$ & 2723.4 & 2813.3 & 2742.9 \\
\bottomrule
\end{tabular}
\label{tab:ControlCombined}
\end{table}


\section{Sensitivity Analysis}
\label{sec:Section5}

In this section it is compared how policy $\mathcal{\pi}_N$ proposed in Section~\ref{sec:Section4} performs with respect to the benchmark policy $\mathcal{\pi}_{DP}$ proposed in Section~\ref{sec:Section3}, when some parameters are different from the nominal case at hand.
Our goal is to empirically evaluate the range of validity of the real-time numeric solution, bearing in mind that it is found with the approximation $\dot{p}(t) \approx 0$.
Clearly, the approximation is expected to be more critical in those conditions where $\dot{p}$ is significantly different from zero during the mission.
From \eqref{eq:HamiltonianFull}, the approximated policy $\mathcal{\pi}_N$ is expected to be good when $p(t) \ll \alpha v_b^{nom} Q_b, \forall t$.

Equation~\eqref{eq:AdjointStateFull} suggests that the derivative of the adjoint state depends on parameters $\beta$, $\gamma$ and on the gradients $\nabla_x \sigma_b$, $\nabla_x P_f$; since $P_f$ is a function of $P_g=P_m-P_b$, it is affected by battery parameters and by the engine efficiency.
For the sake of simplicity, the remainder of the paper focuses on the sensitivity to $\gamma$ and $A_b$.
The cost coefficient $\gamma$ is proportional to the fuel cost, which can significantly change over the time and from country to country.
The parameter $A_b$ is the gradient of the open circuit voltage with respect to the state of charge, which can significantly change from battery to battery, and possibly also over the life of the battery itself.
Nonetheless, it should be noted that the range considered for both the parameters is much larger than a realistic variability and is taken into account only for scientific analysis.
\startmodif
This sensitivity study can clearly be extended to other parameters and to the analytical policy $\pi_X$.
The study is here necessarily limited and leaves possible extensions to future research.
\stopmodif

The sensitivity of the state of charge variation $q_b(0)-q_b(T)$ to $\gamma$ and $A_b$ is shown in Fig.~\ref{fig:SensitivityXUrban}; the corresponding sensitivity of the adjoint state variation $p(0)-p(T)$ is given in Fig.~\ref{fig:SensitivityPUrban}. 
The variation of state of charge depends primarily on $\gamma$, while the dependence on $A_b$ is minor.
\startmodif
Notice that this result also shows how the performance summarized in Table~\ref{tab:BenchmarkUrban} is affected by these parameters: in particular, high values of $\gamma$ make the electric power more convenient.
In Section~\ref{sec:Section3}, the state of charge variation corresponding to the Full Electric strategy is found to be 0.055 (shown in Fig.~\ref{fig:SensitivityPUrban} with a gray dashed line), which is also the variation attained by the TCMS when the fuel is expensive enough.
\stopmodif
The variation of the adjoint state is negligible for small values of $\gamma$, which correspond to policies that tend to increase the state of charge; it tends to increase significantly when both $\gamma$ and $A_b$ increase, and the maximum is attained for their extreme values ($\gamma=0.23$\euro/kWh, $A_b=140$).

\begin{figure}
\centering
\psfragfig{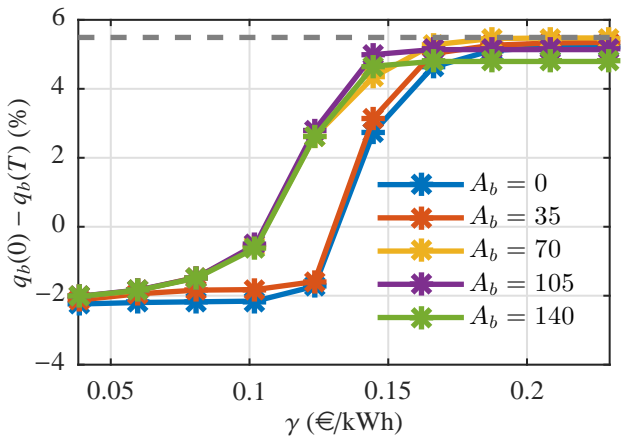}
\caption{ \startmodif Battery state of charge variation over the Urban Driving Cycle for different values of the fuel cost $\gamma$ and the battery parameter $A_b$. The gray dashed line indicates the state of charge variation with the Full Electric strategy. \stopmodif }
\label{fig:SensitivityXUrban}
\end{figure}

\begin{figure}
\centering
\psfragfig{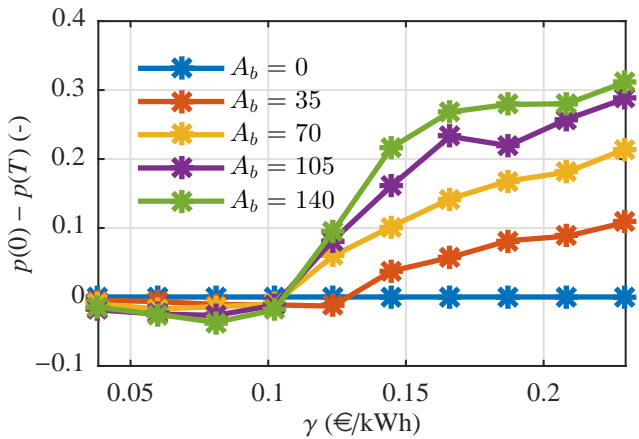}
\caption{ \startmodif Adjoint state variation over the Urban Driving Cycle for different values of the fuel cost $\gamma$ and the battery parameter $A_b$. \stopmodif}
\label{fig:SensitivityPUrban}
\end{figure}

The benchmark solution and the numeric real-time solution are compared in Fig.~\ref{fig:ControlSensUrban} and in Table~\ref{tab:ControlSensUrban}, for the Urban Driving Cycle and for two different choices of $\gamma$ and $A_b$:
\begin{enumerate}
\item[(a)] $\gamma=0.08$\euro/kWh, $A_b=70$, \textit{i.e.}, the nominal parametrization employed in the previous section;
\item[(b)] $\gamma=0.23$\euro/kWh, $A_b=140$, \textit{i.e.}, the parametrization which attains the maximum variation of the adjoint state $p(0)-p(T)$.
\end{enumerate}
The case (a) was already commented in Section~\ref{sec:Section4}. As for the case (b), although the costate variation is significant, the effect on the policy is negligible; the policy $\mathcal{\pi}_{N}$ gives the same results as the benchmark policy $\mathcal{\pi}_{DP}$ in terms of cost, state of charge and fuel consumption.

\begin{figure}
\centering
\includegraphics[]{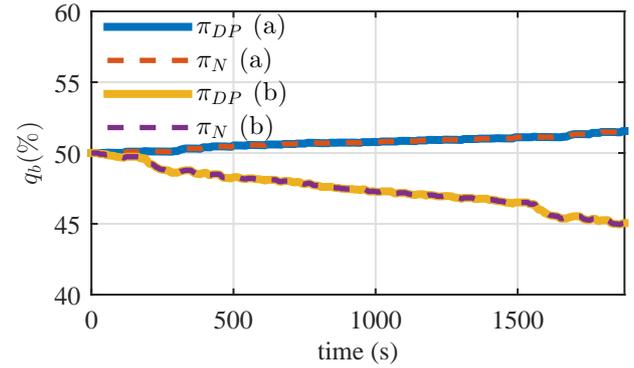}
\caption{Sensitivity analysis for the state of charge in the Urban Driving Cycle: benchmark policy $\mathcal{\pi}_{DP}$ and numerical policy $\mathcal{\pi}_{N}$ in cases (a) and (b).}
\label{fig:ControlSensUrban}
\end{figure}

\begin{table}
\centering
\caption{Sensitivity analysis for the driving cost and the fuel consumption in the Urban Driving Cycle: benchmark policy $\mathcal{\pi}_{DP}$ and numerical policy $\mathcal{\pi}_{N}$ in cases (a) and (b).}
\begin{tabular}{ccccc}
\toprule
 & $\mathcal{\pi}_{DP}$ (a) & $\mathcal{\pi}_{N}$ (a)  & $\mathcal{\pi}_{DP}$ (b) & $\mathcal{\pi}_{N}$ (b)\\
\midrule
Driving Cost $[\text{\euro}]$ & 0.566 & 0.566 & 0.8120 & 0.8121 \\
Fuel Mass $[g]$ & 484.28 & 484.30 & 0 & 0 \\
\bottomrule
\end{tabular}
\label{tab:ControlSensUrban}
\end{table}

The same study was carried out also for the Combined Driving Cycle.
The sensitivity of the state of charge variation $q_b(0)-q_b(T)$ to $\gamma$ and $A_b$ is shown in Fig.~\ref{fig:SensitivityXCombined}; the corresponding sensitivity of the adjoint state variation $p(0)-p(T)$ is given in Fig.~\ref{fig:SensitivityPCombined}. 
The dependence of both $q_b(0)-q_b(T)$ and $p(0)-p(T)$ on $\gamma$ and $A_b$ is almost the same observed for the Urban Driving Cycle. 
The most remarkable difference is the amplitude of the variations. More in detail, notice that if $\gamma$ is high enough, the optimal policy completely depletes the battery and $q_b(0)-q_b(T) = 0.3$.
\startmodif
Also in this case, that this result also shows how the performance summarized in Table~\ref{tab:BenchmarkCombined} is affected by $A_b$ and $\gamma$.
\stopmodif
As for the variation of the adjoint state, it is maximum for the extreme values ($\gamma=0.23$\euro/kWh, $A_b=140$) as in the previous case.

The benchmark policy $\mathcal{\pi}_{DP}$ is compared to policy $\mathcal{\pi}_{N}$ in Fig.~\ref{fig:ControlSensCombined} and in Table.~\ref{tab:ControlSensCombined}, for the Combined Driving Cycle and for the same choices of $\gamma$ and $A_b$ as before.
The first case was already commented in Section~\ref{sec:Section4}.
As for the second case, a significant difference between the two policies is observed.
The increase of both the final cost and the final fuel consumption is of about 3.5\%, whereas the difference between the trends of the state of charge is more relevant: the benchmark policy reaches the lower bound only at the end of the cycle, while policy $\mathcal{\pi}_{N}$ leads to a longer part of the mission spent at the lower bound for the state of charge.
\startmodif
This deterioration of performance can be attributed to the underlying approximation $\dot{p} \approx 0$, that is contradicted by the large variation of the adjoint state observed in Fig.~\ref{fig:SensitivityPCombined}.

In conclusion, the proposed policy $\pi_N$ shows good performance in almost all the tested situations.
The only critical situation involves a long driving cycle and a combination of high fuel cost and high dependence of battery voltage on state of charge, which can falsify the approximations on the adjoint state and consequently make policy $\pi_N$ suboptimal.

\begin{figure}
\centering
\psfragfig{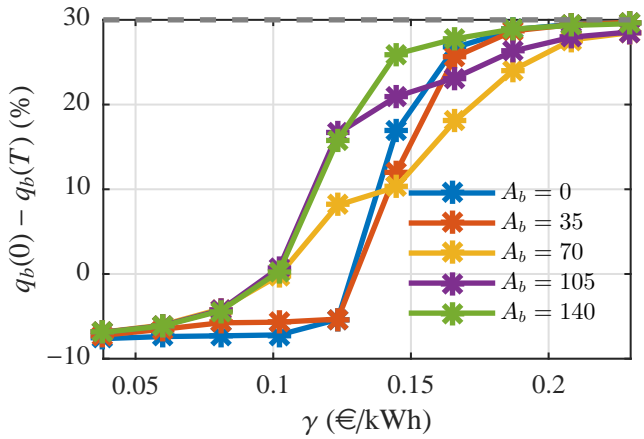}
\caption{ \startmodif Battery state of charge variation over the Combined Driving Cycle for different values of the fuel cost $\gamma$ and the battery parameter $A_b$. The gray dashed line indicates the state of charge variation with the Full Electric strategy. \stopmodif}
\label{fig:SensitivityXCombined}
\end{figure}

While such a combination of parameters is quite extreme at present, notice that the performance degradation could be limited by estimating online the optimal value of the adjoint state, in a similar way to what is done for the ECMS strategies; this amounts to dropping the approximation $p(t)\approx0$ and introducing an estimate $p(t)\approx\hat{p}(t)$.
Notice that the accuracy of the estimate $\hat{p}$ would reasonably depend on the availability of some information on the trip (like elevation and traffic).
This task is, however, outside the scope of this paper and is left for future research.
\stopmodif


\begin{figure}
\centering
\psfragfig{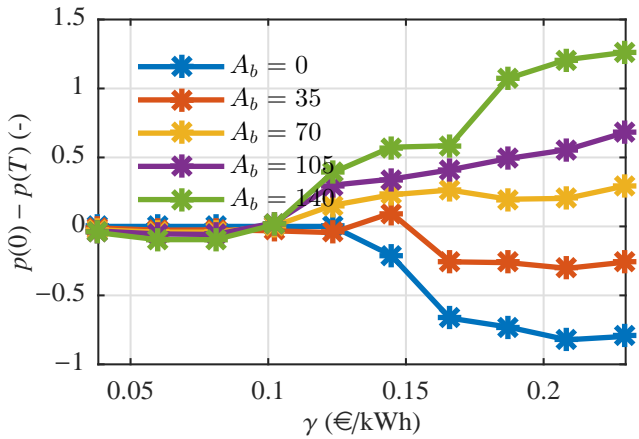}
\caption{ \startmodif Adjoint state variation over the Combined Driving Cycle for different values of the fuel cost $\gamma$ and the battery parameter $A_b$. \stopmodif }
\label{fig:SensitivityPCombined}
\end{figure}


\begin{figure}
\centering
\includegraphics[]{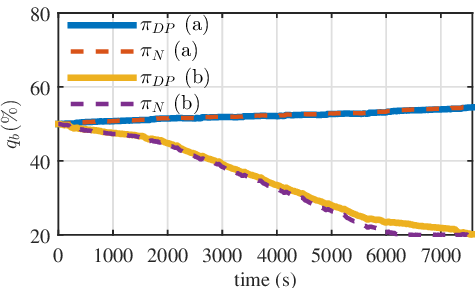}
\caption{Sensitivity analysis for the state of charge in the Combined Driving Cycle: benchmark policy $\mathcal{\pi}_{DP}$ and numerical policy $\mathcal{\pi}_{N}$ in cases (a) and (b).}
\label{fig:ControlSensCombined}
\end{figure}

\begin{table}
\centering
\caption{Sensitivity analysis for the driving cost and the fuel consumption in the Combined Driving Cycle: benchmark policy $\mathcal{\pi}_{DP}$ and numerical policy $\mathcal{\pi}_{N}$ in cases (a) and (b).}
\begin{tabular}{ccccc}
\toprule
 & $\mathcal{\pi}_{DP}$ (a) & $\mathcal{\pi}_{N}$ (a)  & $\mathcal{\pi}_{DP}$ (b) & $\mathcal{\pi}_{N}$ (b)\\
\midrule
Driving Cost $[\text{\euro}]$ & 3.070 & 3.070 & 4.975 & 5.145 \\
Fuel Mass $[g]$ & 2723.4 & 2742.9 & 204.6 & 254.3 \\
\bottomrule
\end{tabular}
\label{tab:ControlSensCombined}
\end{table}


\section{Conclusions and Remarks}
\label{sec:Conclusions}

In this paper, the Total Cost Minimization Strategy (TCMS) was proposed as a suitable approach for least costly energy management in extended range EVs and, more generally, for series HEVs. 

Specifically, the optimization goal was formulated as the overall cost given by the cost of the grid energy, the battery life and the fuel consumption over a given trip. This choice appears more suitable than standard (constrained) fuel minimization, according to the future scenario where more and more HEVs will be plug-in and characterized by high-capacity batteries. Moreover, the use of a ``monetary cost" instead of the standard ``energy consumption" allows to sum up heterogeneous terms without the need of tuning some weighting coefficients. Together with the optimization problem, the model of the powertrain was modified accordingly, in order to take into account the effect of battery aging.

By means of simulations on a full-fledged model of the vehicle, the least costly policy was compared with other policies minimizing different objective functions, when a CNG range extender is available. Here, it was shown that, with current energy costs, the least costly policy does not lead to a full-electric policy even when the driving cycle is within the all-electric range.

Since the benchmark optimum computed using Dynamic Programming cannot be implemented without the \textit{a-priori} knowledge of the driving cycle, a real-time solution of the problem based on the Pontryagin Minimum Principle was studied. In this framework, some analytical guidelines and the numerical solution were provided. It was finally shown that such a real-time strategy is an excellent approximation of the benchmark result for any reasonable combination of model parameters/energy costs.

Future works will be devoted to the implementation of the proposed TCMS approach on a real-world vehicle setup.


\appendix

\section*{Acknowledgements}
This work was partially supported by ``Ricerca sul Sistema Energetico" RSE S.p.a., Milano, Italy.


\section*{References}
\bibliographystyle{elsarticle-num} 
\bibliography{mybibliography}


\section{Simulation Oriented Model Details}
\label{sec:AppendixA}

As described in Section~\ref{sec:Section2}, the simulator is made of a vehicle dynamics part, a battery-related part and a thermal generation area.

The vehicle dynamics area is represented in Fig~\ref{fig:Simulator} by the blocks related to \textit{Vehicle}, \textit{Transmission} and \textit{Motor}.
The block \textit{Vehicle} describes the relationships among the vehicle's longitudinal speed $v$, the slope $\theta$ and the wheel's rotational speed $\omega_w$ and torque $T_w$ as
\begin{gather}
M \dot{v} = \frac{ T_{w} }{ R_{w} } - F_{b} - F_{f} \nonumber\\
\omega_{w} = \frac{v}{R_{w}}\nonumber
\end{gather}
where $M$ is the vehicle mass, $R_w$ is the wheel radius, $F_b$ is the braking force. The friction term $F_f$ can be detailed as
\begin{equation}
F_{f} = - M g \sin{\theta} - C_{r} M g \cos{\theta} - C_{v} v - \frac{1}{2} \rho A C_{x} v^{2}\nonumber
\end{equation}
where $g$ is the gravitational acceleration, $C_r, C_v, C_x$ are respectively the vehicle's roll, viscous and drag coefficients, $\rho$ is the air density, $A$ is the vehicle's reference area.

The block \textit{Transmission} defines the relation that links the wheel's speed $\omega_w$ and wheel torque $T_w$ to the traction motor speed $\omega_m$ and the motor torque $T_m$ as
\begin{gather}
T_{m} = \frac{\eta_{t}^{-\text{sign}(T_w)}}{r} T_{w} \nonumber\\
\omega_{m} = r \omega_{w}\nonumber
\end{gather}
where $r$ is the transmission ratio and $\eta$ is the transmission efficiency.

The block \textit{Motor} models the traction motor power as
\begin{equation}
P_{el} = T_{m} \omega_{m} \eta_{m}(T_{m}, \omega_{m})^{-\text{sign}(P_m)} \nonumber\\
\end{equation}
where $\eta_m$ is the motor efficiency, depending on the mechanical operating point.

As for the battery and thermal generation areas, the underlying model is mainly described by \eqref{eq:Link},~\eqref{eq:BatteryCircuit},~\eqref{eq:BatteryOCVAffine},~\eqref{eq:BatterySoc},~\eqref{eq:BatterySoh},~\eqref{eq:BatterySigmaVariable},~\eqref{eq:EngineFuelRateMap} presented in Section~\ref{sec:Section3}. Equation~\eqref{eq:BatteryCapacityConstant} is replaced by
\begin{equation}
Q_{b}(t) = Q_{b}^{nom} (1 - 0.2 \xi_{b}(t)).\nonumber
\label{eq:BatteryCapacityVariable}
\end{equation}

The numerical values of parameters used in the Simulator and in the Control Oriented Models are given in Table~\ref{tab:ModelParams}.

\begin{table}
\centering
\caption{Simulator and control oriented model parameters}
\begin{tabular}{rc|rc|rc}
\toprule
$A_b \, (\si{\volt})$ & 70 & $B_b \, (\si{\volt})$ & 320 & $v_b^{nom} \, (\si{\volt})$ & 355 \\
$Q_{b}^{nom} \, (\si{Ah})$ & 65 & $N_b \, (-)$ & 2000 & $R_b \, (\si{m\ohm})$ & 500 \\
$A_r \, (-)$ & 3.43 & $B_r \, (\si{k\watt})$ & 5.61 & & \\
$\lambda_{r} \, (\si{k\joule/\g})$ & 47 & $\rho_{f} \, (\si{\kg/l})$ & 0.2 & & \\
$\alpha \, (\frac{\si{\euro}}{\si{kWh}})$ & 0.2 & $\bar{q}_{min} \, (-)$ & 0.2 & $\bar{q}_{max} \, (-)$ & 0.9 \\
$\beta \, (\frac{\si{\euro}}{\si{kWh}})$ & 500 & $\bar{P}_{b,min} \, (\si{kW})$ & -50 & $\bar{P}_{b,max} \, (\si{kW})$ & 50 \\
$\gamma \, (\frac{\si{\euro}}{\si{kWh}})$ & 0.077 & $\bar{P}_{r,min} \, (\si{kW})$ & 0 & $\bar{P}_{r,max} \, (\si{kW})$ & 25 \\
$M \, (\si{kg})$ & 1500	& $C_{x} \, (-)$ & 0.22 & $\rho \, (\si{kg/\cubic\m})$ & 1.18\\
$R_{w} \, (\si{m})$	& 0.3 & $C_{v} \, (\si{kg\per\s})$ & 0 & $r \, (-)$ & 3.5\\
$A \, (\si{\square\m})$ & 2 & $C_{r} \, (-)$ & 0.008 & $\eta_{t} \, (-)$ & 0.98\\
\bottomrule
\end{tabular}
\label{tab:ModelParams}
\end{table}


\section{Power Limits of the Explicit Control Law}
\label{sec:AppendixB}

The power limits separating the regions of the optimal control law of Theorem 1 are defined as
\begin{equation}
\begin{split}
P_{lim}^{(1,4)} &= - ( \Psi_{13} - 4 N_b^{1.5} Q_b^{0.5} p \Psi_1 + \Psi_7 + \Psi_8 + 4 N_b^{1.5} Q_b^{1.5} \alpha v_{oc} \Psi_1\\
&- \Psi_{11} - \Psi_9 + 6 N_b Q_b^2 \sigma_b \alpha \beta v_{oc}^2 + \Psi_{10} - 6 N_b Q_b \sigma_b \beta p v_{oc} \\
&+ N_b^{0.5} Q_b^{1.5} \sigma_b \beta v_{oc} \Psi_1 4 - 4 A_r N_b Q_b^2 \sigma_b \beta \gamma v_{oc}^2 )/\Psi_2 \\
P_{lim}^{(2,4)} &= - (\Psi_{13} + N_b p \Psi_6 + \Psi_7 + \Psi_8 - N_b Q_b \alpha v_{oc} \Psi_6 - Q_b \sigma_b \beta v_{oc} \Psi_6\\
&- \Psi_{11} + \Psi_{12} + 2 A_r B_r N_b^2 Q_b^2  R_b \gamma^2 - \Psi_{14} - A_r N_b^2 Q_b^2 \alpha \gamma v_{oc}^2 \\
&+ A_r N_b^2 Q_b \gamma p v_{oc} - A_r N_b Q_b^2 \sigma_b \beta \gamma v_{oc}^2 )/(2 A_r^2  N_b^2 Q_b^2 R_b \gamma^2 ) \\
P_{lim}^{(3,4)} &= - (\Psi_{13} + \Psi_7 + \Psi_8 - \Psi_{11} - \Psi_9 + \Psi_{12} + \Psi_{10} - \Psi_{14}\\
&+ \Psi_5 - \Psi_3 - \Psi_4 ) / \Psi_2\\
P_{lim}^{(1,5)} &= - ( \Psi_{13} + \Psi_7 + \Psi_8 - \Psi_{11} - \Psi_9 - \Psi_{12} + \Psi_{10} + \Psi_{14}\\
&- \Psi_5 + \Psi_3 - \Psi_4 ) / \Psi_2\\
\end{split}
\label{eq:PowerLimits}
\end{equation}
where
\begin{equation*}
\begin{split}
\Psi_1    &= \sqrt{Q_b \sigma_b \alpha \beta v_{oc}^2  - \sigma_b \beta p v_{oc} + A_r B_r N_b Q_b R_b \gamma^2  - A_r Q_b \sigma_b \beta \gamma v_{oc}^2 } ,\\
\Psi_2    &= 4 A_r^2  N_b^2  Q_b^2  R_b \gamma^2 , \; \Psi_3 = 4 \sqrt{A_r B_r R_b} N_b^2  Q_b^2 \alpha \gamma v_{oc} ,\\
\Psi_4    &= 4 \sqrt{A_r B_r R_b} N_b Q_b^2 \sigma_b \beta \gamma v_{oc} , \; \Psi_5 = 4 \sqrt{A_r B_r R_b} N_b^2  Q_b \gamma p ,\\
\Psi_6    &= (\Psi_{13} + \Psi_7 + \Psi_8 - \Psi_{11} + \Psi_9 + \Psi_{12} + \Psi_{10} - \Psi_{14} \\
&- 2 A_r N_b^2  Q_b^2  \alpha \gamma v_{oc}^2  + 2 A_r N_b^2  Q_b \gamma p v_{oc} - 2 A_r N_b Q_b^2  \sigma_b \beta \gamma v_{oc}^2 )^{0.5} ,\\
\Psi_7    &= N_b^2  Q_b^2  \alpha^2  v_{oc}^2 , \; \Psi_8 = Q_b^2  \sigma_b^2  \beta^2  v_{oc}^2 ,\\
\Psi_9    &= A_r^2  N_b^2  Q_b^2  \gamma^2  v_{oc}^2 , \; \Psi_{10} = 4 A_r B_r N_b^2  Q_b^2  R_b \gamma^2 ,\\
\Psi_{11} &= 2 N_b^2  Q_b \alpha p v_{oc} , \; \Psi_{12} = 2 N_b Q_b^2  \sigma_b \alpha \beta v_{oc}^2 ,\\
\Psi_{13} &= N_b^2  p^2 , \; \Psi_{14} = 2 N_b Q_b \sigma_b \beta p v_{oc}.
\end{split}
\end{equation*}

\startmodif
\section{DP formulation for fuel-based cost functions}
\label{sec:AppendixC}

The ECMS, CS, FE formulations defined in Section~\ref{sec:Section2} are compared with the TCMS formulation in Section~\ref{sec:Section3}; to ensure a fair comparison, all approaches are implemented with DP, using the same underlying model (i.e. the discretized equations \eqref{eq:DiscreteQb}, \eqref{eq:DiscretePf}).
Using the same notation of Section~\ref{sec:Section3}, the discretized running cost for ECMS, CS, FE is
\begin{equation}
g(x(k),u(k),w(k)) = \frac{m_{f}(k+1)-m_{f}(k)}{T_{s}}.
\end{equation}
The constraints on the input $u(k) \in U$ and on the state of charge $x(k) \in X$ are defined in \eqref{eq:TCMSInputSet} and \eqref{eq:TCMSStateSet}, respectively.
The terminal cost is simply $h(x(k),u(k),w(k)) = 0$ for FE and CS, while for the ECMS it is
\begin{equation}
h(x(k),u(k),w(k)) = \zeta (q_b(N)-q_b(0)).
\end{equation}
The final state constraint for CS is $x(N) \in \left\lbrace x(0) \right\rbrace$, while for ECMS and CS it is simply $x(N) \in X$.
\stopmodif




\end{document}